\documentclass[twocolumn]{article}
\usepackage[utf8]{inputenc}
\usepackage{latexsym,amsmath}
\usepackage{amsbsy}
\usepackage{graphicx}
\usepackage{epstopdf}
\usepackage[dvipsnames]{xcolor}
\usepackage{tikz}
\usetikzlibrary{positioning,calc,fit,backgrounds,arrows.meta,decorations.pathreplacing}
\usepackage{float}
\usepackage{booktabs}
\usepackage{array}
\usepackage{authblk}
\usepackage{titletoc}
\usepackage{cite} % sort and compress citation lists into ranges
\usepackage{hyperref}
\usepackage{xurl} % let bibliography URLs break anywhere in narrow columns
\usepackage{physics}
\usepackage{amsthm}
\usepackage{amssymb}
\usepackage{amsfonts}
\usepackage{mathtools}
\usepackage{algorithm}
\usepackage[noend]{algpseudocode}
\usepackage{braket}
\usepackage[a4paper, total={7in, 9.5in}, columnsep=0.25in]{geometry}
\usepackage{fancyhdr}
\newtheorem{theorem}{Theorem}
\newtheorem{lemma}[theorem]{Lemma}

\newtheorem{definition}{Definition}

\ExplSyntaxOn
\NewDocumentCommand{\code}{m}
  {
    \tl_set:Nn \l_tmpa_tl {#1}
    \tl_replace_all:Nnn \l_tmpa_tl {.} {.\allowbreak}
    \texttt{ \tl_use:N \l_tmpa_tl }
  }
\ExplSyntaxOff
\newcommand{\supp}{\operatorname{supp}}      % support of a monomial
\newcommand{\anti}{\operatorname{Anti}}      % anticommuting subset of the tracked monomials
\newcommand{\com}{\operatorname{Com}}       % commuting subset of the tracked monomials
\newcommand{\sites}{\operatorname{sites}}    % the sites a support occupies
\newcommand{\invcount}{\operatorname{inv}}   % inversions of the merged supports
\renewcommand{\rank}{\operatorname{rank}}    % rank assignment (physics.sty defines \rank)
\newcommand{\splitmix}{\operatorname{splitmix}} % within-rank partition hash
\newcommand{\dg}[1]{\lVert #1 \rVert}        % degree, the truncation functional
\newcommand{\gdeg}{g}                        % degree of a pool of gate generators
\newcommand{\ilist}{\mathcal{I}}             % the anticommutation index
\newcommand{\imax}{\ilist_{\max}}            % longest inverted list
\newcommand{\qs}[1]{\hat{#1}}                % anticommutation support: the positions a generator reads
\newcommand{\iu}{\mathrm{i}}                 % imaginary unit
\newcommand{\symd}{\mathbin{\triangle}}      % symmetric difference of sets
\newcommand{\bigsymd}{\mathop{\triangle}\limits} % its n-ary form
\newcommand{\flip}[2]{#1^{#2}}               % occupation string with a mode set flipped

\title{monoprop: Scaling Pauli and Majorana propagation via Anticommutation Indices}
\author[1,2,*]{Aaron Miller}
\affil[1]{Algorithmiq S.r.l., Via della Chiusa 15, 20123 Milan, Italy}
\affil[2]{School of Physics, Trinity College Dublin, College Green, Dublin 2, Ireland}
\date{\vspace*{-1.em}\small (Dated: \today)}

\begin{document}
\makeatletter
\twocolumn[
  \begin{@twocolumnfalse}
    \maketitle
    \begin{abstract}
      We introduce \code{monoprop}, an open-source, massively parallel operator-propagation framework for fermion and qubit bases.
      For \(N\) qubits or modes, existing codes represent each of the \(K\) tracked operator terms (Pauli strings or products of Majorana operators) as a \(2N\)-bit string, testing each for anticommutation with the gate at \(O(KN)\) cost per gate.
      Here, we introduce the \emph{anticommutation index} which rapidly looks up the terms anticommuting with a gate.
      A local gate changes only the terms that touch its sites, and when terms are truncated to act on few sites, only a fraction \(\sim1/N\) touch any given qubit or mode, and we use this to prove the time per gate scales as \(O(K/N)\).
      This counterintuitive result tells us that for a fixed number of terms, increasing the number of sites can make things faster.
      We demonstrate this numerically, and show \code{monoprop} scales almost ideally to 8192 cores, tracking \(9.5\times10^{10}\) terms.
      Beyond propagation, we present memory-efficient surrogate graphs, which record one propagation and replay it at new angles, for fast energy re-estimation and parameter gradients, and an efficient method for finding energy-minimising gates and optimal angles for simulating adaptive quantum-chemistry circuits.
    \end{abstract}
    \vspace{1.5em}
  \end{@twocolumnfalse}
]
\begingroup
\renewcommand\thefootnote{*}
\let\@footnotetext\H@@footnotetext
\footnotetext[1]{\href{mailto:aaron.miller@algorithmiq.fi}{aaron.miller@algorithmiq.fi}}
\endgroup
\makeatother

\section{Introduction}

Operator propagation methods classically approximate expectation values of quantum systems by evolving operators rather than states.
Rather than evolving the state forwards, they expand the observable in an operator basis, evolve it backwards through the circuit in the Heisenberg picture, and evaluate it against the initial state.
Initially introduced for qubits, the approach is known as Pauli propagation or sparse Pauli dynamics~\cite{rall_pauli_propagation_2019, rudolph_pauli_propagation_2025, begusic_fast_2024, begusic_clifford_2025, gharibyan_pps_guide_2025}.

Propagation methods came to prominence by classically simulating large quantum-hardware experiments.
When a 127-qubit superconducting processor was used to compute expectation values from a kicked Ising model beyond exact classical simulation~\cite{kim_evidence_2023}, propagation methods reproduced those observables in a fraction of the quantum runtime~\cite{begusic_spd_2023, begusic_fast_2024,lowesa_2023}, alongside tensor-network and other classical simulations~\cite{tindall_eagle_2024, liao_pepo_2023, patra_tensor_2024, anand_zne_2023, kechedzhi_effective_2024}.
More recently, these methods have been extended to fermionic bases as Majorana propagation~\cite{majorana_propagation_2025} and applied to dynamics of two-dimensional lattice fermions~\cite{danna_majorana_strings_2026}, thermal states through imaginary-time evolution~\cite{rudolph_thermal_2026}, and adaptive construction of ground-state circuits~\cite{adapt-vmpe}.
Truncation errors are provably small for unstructured fermionic circuits, where they decay exponentially with the length of Majorana operators~\cite{majorana_propagation_2025}, for time-evolved observables of weakly interacting fermions~\cite{facelli_fast_convergence_2026}, and for noisy qubit circuits~\cite{aharonov_noisy_2023, schuster_noisy_2025, gonzalez-garcia_pauli_2025}.

These advances have produced open-source, mostly CPU-based, implementations.
\code{PauliPropagation.jl}~\cite{rudolph_pauli_propagation_2025} is a \code{Julia} package seen as the reference implementation for Pauli propagation.
\code{PauliStrings.jl}~\cite{loizeau_paulistrings_2025} and \code{spd}~\cite{begusic_spd_code} target many-body dynamics, \code{pauli-prop}~\cite{qiskit_pauliprop_2025}, \code{PauLIB}~\cite{krotz_paulib_2026} and \code{ppvm}~\cite{ppvm} provide Pauli propagation backends, and \code{MajoranaPropagation.jl}~\cite{danna_majorana_strings_2026} implements Majorana propagation.
GPU Pauli propagation simulators have shown orders of magnitude speedups over CPU codes~\cite{nvidia_cupauliprop, nvidia_cuquantum_2511_blog, mazumder_benchmarking_2026}.
To our knowledge, distributed-memory implementation using the Message Passing Interface (MPI) currently holds the title of the largest simulation, tracking an impressive one trillion Pauli strings over \(12{,}288\) nodes~\cite{broers_orqa_2025}.

Most of these implementations are \emph{dense}, that is, they store each Pauli or Majorana term as a bit string with two bits per qubit or mode, and test anticommutation of every term against every gate.
For \(K\) terms on \(N\) qubits or modes, each gate therefore costs \(O(KN)\) time and memory, although it changes only the terms it anticommutes with.
These implementations do not exploit the sparsity of the basis in which the terms are expressed, as simulations discard terms acting on more than \(d\) sites, nor the locality of gates, which act on \(O(1)\) sites (\emph{local} gates).

Many general purpose quantum software packages instead store each term sparsely, by its few non-identity sites~\cite{mcclean_openfermion_2020, bergholm_pennylane_2018, cirq_2025}.
To our knowledge, however, no high-performance propagation code does, as sparse storage alone makes propagation slower rather than faster.
The essential step of propagation is finding the terms a gate anticommutes with, and without a way to find them directly, every term must still be tested against every gate, with each test much slower than the bitwise operations of a dense code.

This motivates \code{monoprop}: an open-source, MPI-distributed \code{C++} code for Pauli and Majorana propagation, with Python bindings, that makes sparse storage fast.
It stores each term, which we call a \emph{monomial} (Section~\ref{sec:background}), by its few non-identity operators rather than as a full bit string, so a term on \(d\) sites needs \(O(d)\) memory rather than \(O(N)\).
To find the anticommuting monomials without testing the others, \code{monoprop} keeps an \emph{anticommutation index}: for each of the \(2N\) \emph{positions} (\(X_i\) and \(Z_i\) on a qubit, or the two Majorana operators of a mode), a list of the monomials that occupy it.
Anticommutation is the parity of the overlap between a monomial and a gate, so the anticommuting monomials are read off the few lists the gate selects, and monomials sharing no position with the gate are never touched.
The idea of the index can be summarised by what the reader does when they look for the proofs of this result in the appendices of this paper.
Instead of reading all the details in every section of the appendix, they look at the contents and navigate to the sections of interest.
In this example, you, the reader, are a quantum gate, the sections are the index lists and their subsections are the monomials, so you only read the sections of the positions you act on (care about), and you anticommute with the monomial subsections that interest you.
We note that the reader should examine the appendices carefully.

Figure~\ref{fig:kn-schematic} shows the intuition behind the advantage of this method.
A dense code (bottom) reads every monomial in full, while the index (top) reads only the lists of the positions a gate acts on.
When the system grows, their occupied positions are shared among more lists, so each list, and thus the work of a gate, shrinks, while the dense version grows linearly.
A gate on \(O(1)\) positions therefore costs \(O(K/N)\) rather than \(O(KN)\) when the monomials are spread evenly (Section~\ref{sec:inverted-index}), so the cost per tracked monomial counterintuitively falls as the system grows.
In other words, \code{monoprop} retrieves the monomials a gate anticommutes with, and rotates them, in time proportional to how many there are, rather than to how many are tracked.
In a kicked-Ising chain, for example, a gate anticommutes with \(15\%\) of the tracked monomials at \(N=32\) but only \(0.46\%\) at \(N=1024\) (Fig.~\ref{fig:inverse-scaling}(c)).
\code{monoprop} touches only these, and its measured cost per monomial falls close to \(1/N\), whereas a dense code checks \(100\%\) of the monomials at every gate.

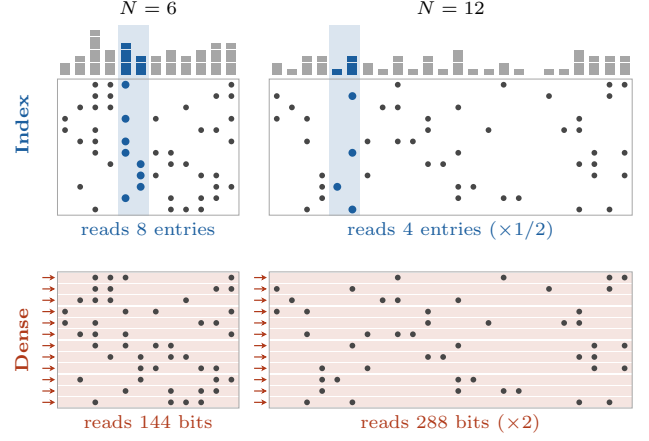
\begin{figure}[t]
  \centering
  %% figures/kn-schematic.tex --- where the K/N comes from.
%% Included from main.tex inside a single-column figure; needs tikz with the
%% arrows.meta library.
%%
%% The same K = 12 degree-4 Majorana monomials (rows; dots mark occupied
%% positions) on N = 6 and N = 12 modes, with a local generator on positions
%% 5 and 6.
%%   top     the anticommutation index: each position has a list (bar above
%%           the column, one block per entry); the gate reads only its two
%%           columns. The 4K entries are shared among 2N lists, so doubling N
%%           halves what the gate reads: 8 -> 4 entries.
%%   bottom  a dense code asks every monomial, reading all 2N bits of each
%%           row: 144 -> 288 bits.
%%
%% Positions are uniformly random (seeded), not balanced, so list lengths vary
%% as in a real operator; the draw was chosen with gate reads equal to their
%% expected value, 2d at N = 6 and d at N = 12.
%% To change the example, edit \rowsA/\rowsB and the matching list lengths
%% \lensA/\lensB.
%%
%% Colours follow inverted-index-majorana.tex: qcol is what the index reads,
%% dcol (acol there) what the dense scan reads. No transparency or background
%% layers: they corrupted the PDF inside repeated macros.

\definecolor{qcol}{HTML}{1B5E9E}
\definecolor{dcol}{HTML}{B03A1A}

\begin{tikzpicture}[
    x=1cm, y=1cm,
    hd/.style  = {font=\scriptsize, inner sep=1pt},
    lab/.style = {font=\scriptsize\bfseries, inner sep=1pt, rotate=90},
    frame/.style = {draw=black!35, line width=0.3pt},
  ]

  % ---------------------------------------------------------------- geometry
  \def\c{0.20}                     % column pitch (one position)
  \def\r{0.15}                     % row pitch (one monomial)
  \def\K{12}
  \def\xa{0.55}                    % left edge of the N = 6 panels
  \def\xb{3.35}                    % left edge of the N = 12 panels
  \def\yi{0}                       % top of the index panels
  \def\yd{-2.55}                   % top of the dense panels
  \def\hp{0.085}                   % pitch of the list-length blocks

  % -------------------------------------------------------------------- data
  \def\rowsA{{3,4,5,12},{3,4,11,12},{2,3,4,9},{1,5,7,12},{1,3,10,11},
    {2,3,5,12},{3,5,7,8},{4,6,8,9},{6,7,10,11},{2,6,11,12},{5,8,10,11},
    {3,8,9,10}}
  \def\lensA{2,3,7,4,5,3,3,4,3,4,5,5}
  \def\rowsB{{9,16,23,24},{1,6,19,23},{2,8,9,13},{1,3,11,21},{11,15,20,21},
    {3,7,9,10},{6,13,22,24},{11,12,21,22},{4,7,22,24},{4,5,13,14},
    {4,13,16,17},{3,6,21,23}}
  \def\lensB{2,1,3,3,1,3,2,1,3,1,3,1,4,1,1,2,1,0,1,1,4,3,3,3}

  %% Dots of one panel: #1 x-origin, #2 y-top, #3 rows, #4 1 = index panel.
  \newcommand{\kndots}[4]{
    \foreach \m [count=\i] in #3 {
        \foreach \p in \m {
            \ifnum#4=1 \ifnum\p=5 \def\hit{1}\else\ifnum\p=6 \def\hit{1}\else
                  \def\hit{0}\fi\fi
            \else \def\hit{0}\fi
            \ifnum\hit=1
              \fill[qcol] ({#1+(\p-0.5)*\c},{#2-(\i-0.5)*\r}) circle (0.052);
            \else
              \fill[black!72] ({#1+(\p-0.5)*\c},{#2-(\i-0.5)*\r}) circle (0.038);
            \fi
          }
      }
  }

  %% Index panel: #1 x-origin, #2 positions, #3 rows, #4 list lengths.
  \newcommand{\knindex}[4]{
    \fill[qcol!16] ({#1+4*\c},{\yi+0.12+7*\hp}) rectangle ({#1+6*\c},{\yi-\K*\r});
    \draw[frame] (#1,\yi) rectangle ({#1+#2*\c},{\yi-\K*\r});
    \foreach \n [count=\j] in #4 {
        \ifnum\j=5 \def\bc{qcol}\else\ifnum\j=6 \def\bc{qcol}\else
              \def\bc{black!35}\fi\fi
        \ifnum\n>0
          \foreach \k in {1,...,\n}
          {\fill[\bc] ({#1+(\j-1)*\c+0.035},{\yi+0.06+(\k-1)*\hp})
            rectangle ({#1+\j*\c-0.035},{\yi+0.06+\k*\hp-0.018});}
        \fi
      }
    \kndots{#1}{\yi}{#3}{1}
  }

  %% Dense panel: #1 x-origin, #2 positions, #3 rows.
  \newcommand{\kndense}[3]{
    \foreach \i in {1,...,\K}
    {\fill[dcol!13] (#1,{\yd-(\i-1)*\r-0.015}) rectangle
      ({#1+#2*\c},{\yd-\i*\r+0.015});
      \draw[-{Stealth[length=2.2pt,width=2.2pt]},dcol,line width=0.45pt]
      ({#1-0.20},{\yd-(\i-0.5)*\r}) -- ({#1-0.03},{\yd-(\i-0.5)*\r});}
    \draw[frame] (#1,\yd) rectangle ({#1+#2*\c},{\yd-\K*\r});
    \kndots{#1}{\yd}{#3}{0}
  }

  % ------------------------------------------------------------------ panels
  \knindex{\xa}{12}{\rowsA}{\lensA}
  \knindex{\xb}{24}{\rowsB}{\lensB}
  \kndense{\xa}{12}{\rowsA}
  \kndense{\xb}{24}{\rowsB}

  % ------------------------------------------------------------------ labels
  \node[hd,anchor=south] at ({\xa+6*\c},{\yi+0.82}) {\(N=6\)};
  \node[hd,anchor=south] at ({\xb+12*\c},{\yi+0.82}) {\(N=12\)};

  \node[lab,text=qcol] at (0.08,{\yi-0.5*\K*\r+0.3}) {Index};
  \node[lab,text=dcol] at (0.08,{\yd-0.5*\K*\r}) {Dense};

  \node[hd,text=qcol,anchor=base] at ({\xa+6*\c},{\yi-\K*\r-0.26})
  {reads 8 entries};
  \node[hd,text=qcol,anchor=base] at ({\xb+12*\c},{\yi-\K*\r-0.26})
  {reads 4 entries (\(\times1/2\))};
  \node[hd,text=dcol,anchor=base] at ({\xa+6*\c},{\yd-\K*\r-0.26})
  {reads 144 bits};
  \node[hd,text=dcol,anchor=base] at ({\xb+12*\c},{\yd-\K*\r-0.26})
  {reads 288 bits (\(\times2\))};

\end{tikzpicture}
  \caption{\label{fig:kn-schematic} \textbf{K/N scaling intuition.}
    Twelve degree-4 Majorana monomials (rows) on \(N=6\) and \(N=12\) modes, with dots marking occupied positions.
    In the top panels, the gate reads only the two blue highlighted anticommutation index lists, whose lengths are shown by bars.
    Doubling the number of modes stretches each of the degree-\(4\) monomials among twice as many lists.
    Because the operator is evenly supported, these lists have similar lengths (Definition~\ref{def:even-occupancy}), so the gate now reads about half as many entries.
    By contrast, the dense version below reads everything, meaning we have double the work with double the modes.}
\end{figure}

\begin{figure*}[t!]
  \centering
  \includegraphics[width=0.9\textwidth]{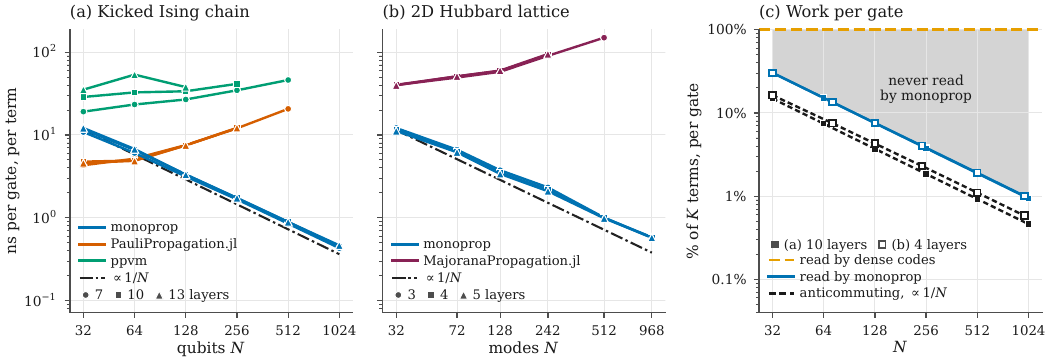}
  \caption{\textbf{Scaling with system size.} Single-core run time divided by \(\sum_g K_g\), the number of monomials held when each gate \(g\) is applied, summed over gates (Appendix~\ref{app:held}), with every code tracking the same number of monomials.
    Both models show the \(O(K/N)\) scaling of Theorem~\ref{thm:xor}.
    (a) Kicked-Ising chain at Pauli-weight cutoff 6 after 7, 10 and 13 layers (\(10^{6}\)--\(4\times10^{8}\) monomials), against the dense codes \code{PauliPropagation.jl}~\cite{rudolph_pauli_propagation_2025} and \code{ppvm}~\cite{ppvm}.
    (b) Two-dimensional Hubbard lattice at cutoff 3 after 3, 4 and 5 layers (\(10^{4}\)--\(10^{8}\) monomials), against the dense code \code{MajoranaPropagation.jl}~\cite{danna_majorana_strings_2026}.
    (c) Proportion of monomials that anticommute with a gate (black dashed), at 10 layers of (a) and 4 of (b). Dense codes test \(100\%\) of the monomials (orange), while \code{monoprop} reads at most about twice the anticommuting set.
    Runs over 30 minutes are omitted.}
  \label{fig:inverse-scaling}
\end{figure*}

The same data structures distribute across thousands of cores.
Monomials are distributed over processes by a hash that is linear under monomial multiplication, so each gate makes every process exchange data with only one other process.
On the Deucalion supercomputer, \code{monoprop} tracks \(9.47\times10^{10}\) monomials with near ideal parallel scaling.

Variational optimisation re-evaluates the same circuit at many angles.
A \emph{surrogate} records the propagation once and replays it cheaply at new angles~\cite{fontana_noisy_2025, lowesa_2023, rudolph_pauli_propagation_2025, majorana_propagation_2025}.
\code{monoprop}'s surrogate graph costs about one propagation to build, then returns an expectation value in \(\sim0.1\times\) and the gradient with respect to every parameter in \(\sim0.2\times\) the propagation time (Section~\ref{sec:graph-gradient}).
It uses the index to reconstruct the anticommuting monomials of each gate quickly, rather than storing them.

Finally, adaptive variational algorithms for quantum chemistry, such as ADAPT-VQE~\cite{grimsley_adapt_2019, feniou_greedy_2025}, grow a circuit one gate at a time, adding the gate and angle from a predefined candidate set (the \emph{pool}) that most lowers the energy.
For a pool of degree-\(\gdeg\) Majorana generators and an evolved Hamiltonian of degree \(d\), scoring every candidate directly costs \(O(N^{d+\gdeg})\).
For Fock reference states, we compute the gradient, optimal angle, and energy decrease of every candidate in \(O(N^{(d+\gdeg)/2})\) time, halving the exponent.

\section{Background}\label{sec:background}

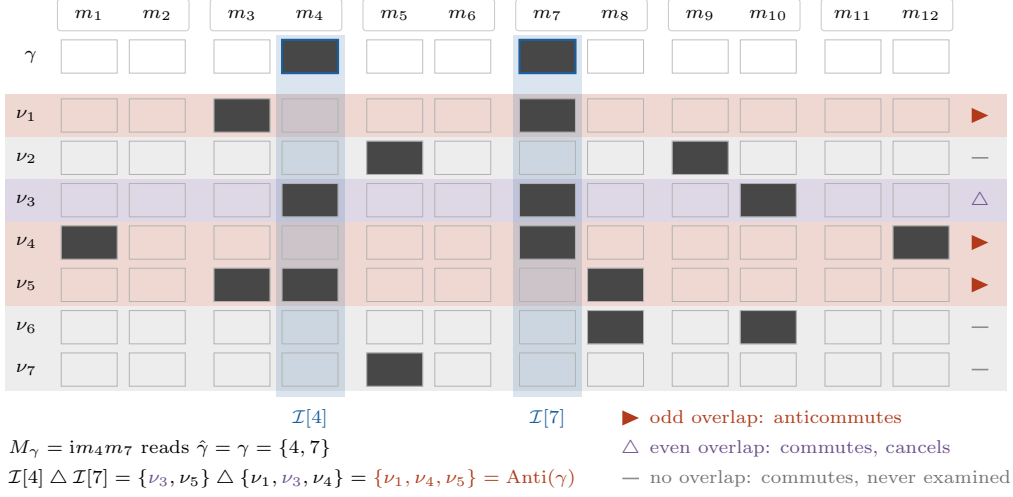
\begin{figure*}[t]
  \centering
  \definecolor{qcol}{HTML}{1B5E9E}   % the lookup: positions read, lists traversed
\definecolor{acol}{HTML}{B03A1A}   % the answer: the anticommuting set
\definecolor{ccol}{HTML}{6B4E8C}   % visited by the lookup, but commutes

\begin{tikzpicture}[
    x=1cm, y=1cm,
    cell/.style  = {draw=black!30, line width=0.25pt},
    onbit/.style = {draw=black!30, line width=0.25pt, fill=black!72},
    gcell/.style = {draw=black!25, line width=0.25pt},
    gbit/.style  = {draw=black!30, line width=0.25pt, fill=black!72},
    qbox/.style  = {draw=qcol, line width=0.8pt},
    hd/.style    = {font=\scriptsize, inner sep=1pt},
    note/.style  = {font=\scriptsize, align=left, inner sep=1pt},
    tile/.style  = {draw=black!20, line width=0.4pt, rounded corners=1.2pt,
        fill=white},
  ]

  % ---------------------------------------------------------------- geometry
  \def\cw{0.90}                    % column pitch
  \def\gp{0.22}                    % gap between neighbouring modes
  \def\rh{0.56}                    % row pitch
  \def\gh{3.92}                    % depth of the grid: 7*\rh
  \def\gy{0.50}                    % centre of the generator row
  \def\lx{-0.66}                   % left edge of the row bands, and of the footer
  \def\vx{12.12}                   % the verdict glyphs
  \def\bx{12.58}                   % right edge of the row bands
  \def\cx#1{(#1-1)*\cw+int((#1-1)/2)*\gp+0.5*\cw}
  \def\cy#1{-((#1-1)*\rh+0.5*\rh)}
  \def\boxat#1#2{({\cx{#2}-0.5*\cw+0.08},{#1-0.5*\rh+0.06})
    rectangle ({\cx{#2}+0.5*\cw-0.08},{#1+0.5*\rh-0.06})}
  \def\cellat#1#2{\boxat{\cy{#1}}{#2}}

  %% The axis is one row deep: the position index doubles as the Majorana
  %% index, since position j simply is m_j in this engine, so m_4 names its own
  %% mode and a row of mode numbers above it would only repeat itself. The
  %% tiles are kept, since they carry the pairing of positions into modes.
  \def\ytop{1.26}                  % top of the tiles
  \def\ypos{1.05}                  % the positions, m_j
  \def\ybot{0.84}                  % foot of the tiles

  %% The footer band: three baselines under the grid, shared by the columnwise
  %% reading on the left and the legend at \kx on the right, so that the two
  %% blocks fill one band rather than hanging separately.
  \def\yla{-4.68}
  \def\ylb{-5.08}
  \def\ylc{-5.48}
  \def\kx{7.50}                    % left edge of the legend glyphs
  %% The legend starts one line higher, level with the list labels, so the
  %% band is two lines deep on the left rather than three.
  \def\kya{-4.28}
  \def\kyb{-4.68}
  \def\kyc{-5.08}

  % -------------------------------------------------------------------- data
  \def\occ{1/{3,7},2/{5,9},3/{4,7,10},4/{1,7,12},5/{3,4,8},6/{8,10},7/{5}}
  \def\gam{4,7}                    % \gamma = m_4m_7, and \qs{\gamma} = \gamma
  \def\hitrows{1,4,5}              % odd overlap: anticommutes, returned
  \def\evenrow{3}                  % even overlap: visited, commutes, cancels
  \def\greyrows{2,6,7}             % no overlap: commutes, never examined

  % ------------------------------------------- row states and columns read
  \begin{scope}[on background layer]
    \foreach \i in \greyrows
    {\fill[black!7] ({\lx},{\cy{\i}-0.5*\rh}) rectangle
        ({\bx},{\cy{\i}+0.5*\rh});}
    \foreach \i in \hitrows
    {\fill[acol!20] ({\lx},{\cy{\i}-0.5*\rh}) rectangle
        ({\bx},{\cy{\i}+0.5*\rh});}
    \foreach \i in \evenrow
    {\fill[ccol!22] ({\lx},{\cy{\i}-0.5*\rh}) rectangle
        ({\bx},{\cy{\i}+0.5*\rh});}
    \foreach \j in \gam {
      \fill[qcol,fill opacity=0.16] ({\cx{\j}-0.5*\cw},{\gy+0.5*\rh})
      rectangle ({\cx{\j}+0.5*\cw},{-\gh-0.08});
    }
  \end{scope}

  % ------------------------------------------------------------- the axis
  \foreach \s in {1,...,6} {
      \draw[tile] ({\cx{2*\s-1}-0.5*\cw+0.03},{\ybot})
      rectangle ({\cx{2*\s}+0.5*\cw-0.03},{\ytop});
      \node[hd] at ({\cx{2*\s-1}},{\ypos}) {\(m_{\the\numexpr2*\s-1\relax}\)};
      \node[hd] at ({\cx{2*\s}},{\ypos})   {\(m_{\the\numexpr2*\s\relax}\)};
    }

  % --------------------------------------------------------- the generator
  %% One strip serves as both the support and the lookup: an even-length
  %% Majorana generator reads exactly the positions it occupies.
  \node[hd,anchor=east] at (-0.20,{\gy}) {\(\gamma\)};
  \foreach \j in {1,...,12} {\draw[gcell] \boxat{\gy}{\j};}
  \foreach \j in \gam {\draw[gbit] \boxat{\gy}{\j};
    \draw[qbox] \boxat{\gy}{\j};}

  % ------------------------------------------------------- operator and matrix
  \foreach \i in {1,...,7} {
      \node[hd,anchor=east] at (-0.20,{\cy{\i}}) {\(\nu_{\i}\)};
      \foreach \j in {1,...,12} {\draw[cell] \cellat{\i}{\j};}
    }
  \foreach \i/\ps in \occ {\foreach \j in \ps {\draw[onbit] \cellat{\i}{\j};}}

  % ----------------------------------------------------------- the verdicts
  \foreach \i in \hitrows
  {\fill[acol] ({\vx},{\cy{\i}-0.10}) -- ({\vx},{\cy{\i}+0.10})
    -- ({\vx+0.22},{\cy{\i}}) -- cycle;}
  \foreach \i in \evenrow
  {\node[hd,text=ccol] at ({\vx+0.11},{\cy{\i}}) {\(\symd\)};}
  \foreach \i in \greyrows
  {\draw[black!45,line width=0.6pt] ({\vx},{\cy{\i}}) --
    ({\vx+0.22},{\cy{\i}});}

  % --------------------------------------------------- the columnwise reading
  \foreach \j in \gam
  {\node[hd,text=qcol,anchor=north] at ({\cx{\j}},{-\gh-0.20})
    {\(\ilist[\j]\)};}

  \node[note,anchor=west] at ({\lx},{\yla})
  {\(M_\gamma=\iu m_4m_7\) reads \(\qs{\gamma}=\gamma=\{4,7\}\)};
  \node[note,anchor=west] at ({\lx},{\ylb})
  {\(\ilist[4]\symd\ilist[7]
    = \{{\color{ccol}\nu_3},\nu_5\}
    \symd \{\nu_1,{\color{ccol}\nu_3},\nu_4\}
    = {\color{acol}\{\nu_1,\nu_4,\nu_5\} = \anti(\gamma)}\)};

  % ------------------------------------------ the legend, right of the footer
  %% Its three lines sit on their own baselines, one line above the footer's, and end flush with the right
  %% edge of the row bands, so the band reads as one block and not as a column
  %% of text hung beside the grid.
  \fill[acol] ({\kx},{\kya-0.10}) -- ({\kx},{\kya+0.10})
  -- ({\kx+0.22},{\kya}) -- cycle;
  \node[hd,anchor=west,text=acol] at ({\kx+0.32},{\kya})
  {odd overlap: anticommutes};
  \node[hd,text=ccol] at ({\kx+0.11},{\kyb}) {\(\symd\)};
  \node[hd,anchor=west,text=ccol] at ({\kx+0.32},{\kyb})
  {even overlap: commutes, cancels};
  \draw[black!45,line width=0.6pt] ({\kx},{\kyc}) -- ({\kx+0.22},{\kyc});
  \node[hd,anchor=west,text=black!55] at ({\kx+0.32},{\kyc})
  {no overlap: commutes, never examined};

\end{tikzpicture}
  \caption{\label{fig:inverted-index} \textbf{Anticommutation index} showing \(K=7\) Majorana monomials on \(N=6\) modes, with black occupied positions.
    The generator \(\gamma=\{4,7\}\) reads the lists of its two positions (shaded), whose symmetric difference is the anticommuting set (Eq.~\eqref{eq:index-query}).
    Monomials \(\nu_{2,6,7}\) occupy neither position, so commute and are never examined, while \(\nu_3\) occupies both, cancels, and commutes.
    Monomials \(\nu_{1,4,5}\) occupy exactly one and anticommute.
    The Pauli index is the same structure and differs only in the lists a generator reads (Fig.~\ref{fig:pauli-query}).}
\end{figure*}

We wish to estimate the expectation value \(\Tr[\varrho\,U^\dagger H U]\) of an observable \(H\) after a circuit \(U=U_L\cdots U_1\) acting on an initial state \(\varrho\).
Propagation evolves \(H\) backwards, first to \(U_L^\dagger H U_L\), then through \(U_{L-1}\), and so on, and evaluates the result against \(\varrho\) only at the end.
To do so, we expand \(H\) in an operator basis, \(H=\sum_\nu c_\nu M_\nu\), with Hermitian \(M_\nu\) satisfying \(M_\nu^2=1\), and track the operator basis and coefficients \(c_\nu\) as each gate is applied.

For qubits, the \(M_\nu\) are Pauli strings, tensor products of \(\{I, X, Y, Z\}\) on \(N\) qubits.
For fermions, they are products of Majorana operators \(\{m_{2i-1}, m_{2i}\}\) on \(N\) modes.
In both bases, each is a product of a few single-site operators, which we call a \emph{monomial}, and we refer to qubits and modes collectively as \emph{sites}.

Each site carries two \emph{positions}: on qubit \(i\), positions \(2i-1\) and \(2i\) stand for \(X_i\) and \(Z_i\) (the standard X/Z encoding), and on mode \(i\), for the Majorana operators \(m_{2i-1}\) and \(m_{2i}\).
A monomial is labelled by the positions it occupies, its \emph{support} \(\nu\subseteq[2N]=\{1,\ldots,2N\}\).
For Pauli strings, occupying both positions of a qubit gives \(Y_i=\iu X_iZ_i\).
Writing \(X(\nu),Z(\nu)\subseteq[N]\) for the positions of qubits with \(X\) or \(Z\) occupied, and \(y(\nu)\) for the number of doubly occupied sites, the Pauli monomial is
\begin{equation}
  M_\nu = \iu^{y(\nu)} \prod_{i\in X(\nu)} X_i \prod_{i\in Z(\nu)} Z_i .
  \label{eq:pauli-monomial}
\end{equation}
For fermions, each occupied position is a Majorana operator, and for \(\nu=\{j_1<\cdots<j_k\}\) the Majorana monomial is
\begin{equation}
  M_\nu = \iu^{k(k-1)/2}\, m_{j_1}m_{j_2}\cdots m_{j_k} .
  \label{eq:majorana-monomial}
\end{equation}
The phase makes these monomials Hermitian (Appendix~\ref{app:monomial-algebra}).
As \(\nu\) determines \(M_\nu\), we also call \(\nu\) a monomial, and use \(M\) in both bases.

The \emph{degree} \(\dg{\nu}\) of a monomial is the number of qubits a Pauli string acts on non-trivially (its weight), or the number of Majorana operators in a Majorana monomial (its length).
A Majorana monomial occupies exactly \(\dg{\nu}\) positions, while a Pauli \(Y_i\) occupies two but counts once, so a monomial of degree \(d\) occupies at most \(2d\) positions.
For example,
\begin{align*}
  Y_1Z_2     & :\quad
  \overbrace{\underset{X_1}{\bullet}\;\,\underset{Z_1}{\bullet}}^{\text{qubit }1}\quad
  \overbrace{\underset{X_2}{\circ}\;\,\underset{Z_2}{\bullet}}^{\text{qubit }2}
  \qquad \nu=\{1,2,4\},\quad \dg{\nu}=2, \\
  \iu m_1m_4 & :\quad
  \overbrace{\underset{m_1}{\bullet}\;\,\underset{m_2}{\circ}}^{\text{mode }1}\quad
  \overbrace{\underset{m_3}{\circ}\;\,\underset{m_4}{\bullet}}^{\text{mode }2}
  \qquad \nu=\{1,4\},\quad \dg{\nu}=2,
\end{align*}
where filled circles are occupied positions.
Both have degree two, but \(Y_1\) fills both positions of qubit 1.

A gate \(U=e^{-\iu \theta M_\gamma/2}\) is generated by a single monomial \(M_\gamma\).
It leaves monomials commuting with \(M_\gamma\) unchanged and rotates anticommuting ones as
\begin{equation}
  c_\nu M_\nu \rightarrow c_\nu \cos(\theta)\, M_\nu \pm c_\nu \sin(\theta)\, M_{\nu\symd\gamma},
\end{equation}
where the symmetric difference \(\nu\symd\gamma\), the \emph{partner} of \(\nu\), labels the product \(M_\gamma M_\nu\propto M_{\nu\symd\gamma}\).
Each rotation can double the number of tracked monomials, so simulations often discard monomials of degree above a cutoff \(d \ll N\), keeping \(K=O(N^d)\) terms, and monomials whose coefficient magnitude falls below a threshold.

A dense implementation usually stores each monomial as a \(2N\)-bit string in \(\lceil 2N/64\rceil\) 64-bit words, so it uses \(O(KN)\) memory and time per gate.
Most of the fastest existing implementations are dense, and we benchmark them in the next section.

\section{Results}\label{sec:results}

\subsection{Anticommutation index}\label{sec:inverted-index}
\begin{figure*}[t]
  \centering
  \includegraphics[width=0.8\textwidth]{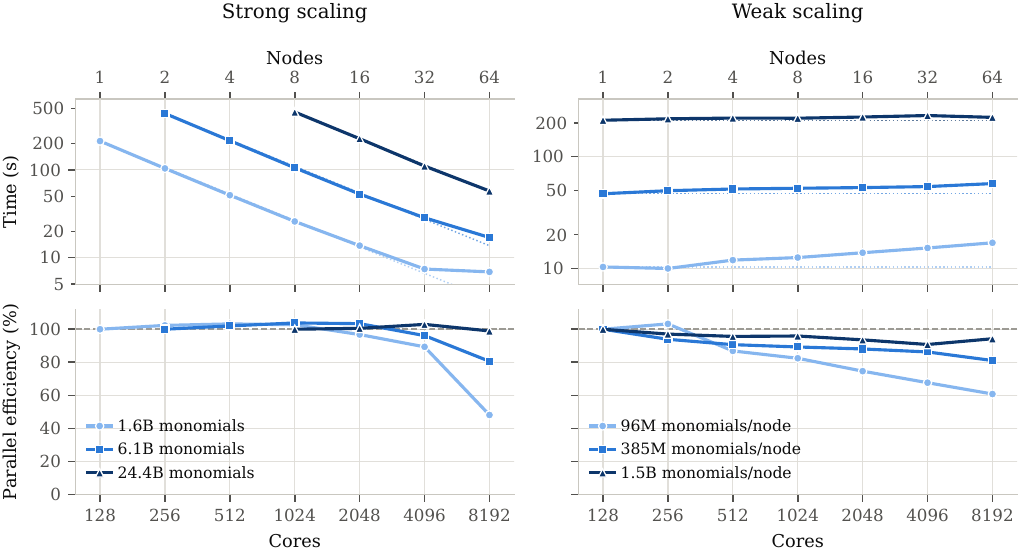}
  \caption{\textbf{Strong and weak scaling of \code{monoprop}.} Distributed propagation of the linear Hubbard model of Appendix~\ref{app:hubbard}, from 128 cores on one node to 8192 cores on 64 nodes.
    The strong scaling column fixes the number of tracked monomials at \(1.6\times10^{9}\), \(6.1\times10^{9}\), or \(2.44\times10^{10}\).
    The weak scaling column holds the monomials per node fixed at \(9.6\times10^{7}\), \(3.85\times10^{8}\), or \(1.48\times10^{9}\).
    The top row shows wall-clock time against the ideal scaling curves, and the bottom row reports the parallel efficiency.
    Larger problems maintain their efficiency to more cores.}
  \label{fig:mpi-scaling}
\end{figure*}

A gate changes only the monomials it anticommutes with, so finding them is the core of propagation, and where dense codes spend their \(O(KN)\) time.
Applying a gate \(M_\gamma\) therefore starts by finding the subset \(\anti(\gamma)\) of tracked monomials anticommuting with it.
Monomials either commute or anticommute with a gate.
Reordering \(M_\nu M_\gamma\) into \(M_\gamma M_\nu\) moves each single-site operator of \(M_\nu\) past each one of \(M_\gamma\), and every pair of single-site operators that anticommute flips the sign, so \(M_\nu M_\gamma=(-1)^{p}M_\gamma M_\nu\), where \(p\) counts these pairs, and they anticommute when \(p\) is odd.
Distinct Majorana operators anticommute, so each swap of Majoranas is a sign flip and \(p=|\nu||\gamma|-|\nu\cap\gamma|\), where \(|\cdot|\) counts occupied positions.
For an even-length generator, this reduces to the parity of the overlap \(|\nu\cap\gamma|\), for example, \(\iu m_1m_2\) and \(\iu m_2m_3\) share a Majorana and anticommute.
For Paulis, \emph{pairs of single-site operators} anticommute only when they act on the same site and differ, as \(X_i\) and \(Z_i\) do, so \(X_1X_2\) anticommutes with \(Z_1\) but commutes with \(Z_1Z_2\).

More generally, anticommutation is the parity of the overlap of \(\nu\) with a set \(\qs{\gamma}\subseteq[2N]\) fixed by the gate and basis, which we call its \emph{anticommutation support}.
For an even-length Majorana generator, \(\qs{\gamma}=\gamma\) (Fig.~\ref{fig:inverted-index}).
For a Pauli generator, \(\qs{\gamma}\) is \(\gamma\) with the \(X_i\) and \(Z_i\) positions exchanged on each site, so that an \(X_i\) in \(\nu\) overlaps a \(Z_i\) or \(Y_i\) in \(\gamma\) (Fig.~\ref{fig:pauli-query}).
For an odd-length Majorana generator, \(p\) has the parity of \(|\nu\setminus\gamma|\), so \(\qs{\gamma}\) is the complement \([2N]\setminus\gamma\) (Appendix~\ref{app:query-support}).
Thus,
\begin{equation}
  \nu \in \anti(\gamma)
  \iff
  \bigl|\nu \cap \qs{\gamma}\bigr| \equiv 1 \pmod 2.
  \label{eq:parity-test}
\end{equation}
Local Pauli and even Majorana rotation gates have \(|\qs{\gamma}|=O(1)\).
Odd Majorana generators have \(|\qs{\gamma}|=\Theta(N)\).

\begin{definition}[Anticommutation index]\label{def:inverted-index}
  Let \(\mathcal K_d\) be a set of \(K\) monomials of degree at most \(d\), so each has Pauli weight or Majorana length at most \(d\).
  The index is the \(2N\) lists
  \[
    \ilist[j] \coloneq \{\,\nu\in\mathcal K_d : j\in\nu\,\},
    \qquad j\in[2N].
  \]
\end{definition}
This is an example of an inverted index from information retrieval~\cite{manning_ir_2008}.
A monomial appears in the lists selected by \(\qs{\gamma}\) once for each overlapping position.
Their symmetric difference \(\symd\) retains the monomials with odd overlap, yielding \(\anti(\gamma)\) by Eq.~\eqref{eq:parity-test}.

\begin{theorem}\label{thm:xor}
  For a gate generator \(M_\gamma\) with anticommutation support \(\qs{\gamma}\), the index returns
  \begin{equation}
    \anti(\gamma) = \bigsymd_{j\,\in\,\qs{\gamma}} \ilist[j]
    \label{eq:index-query}
  \end{equation}
  in \(O(|\qs{\gamma}|\imax)\) time where \(\imax\coloneq\max_j|\ilist[j]|\) is the longest list.
  When \(d, |\qs{\gamma}| \in O(1)\), the lookup costs \(O(K)\) in the worst case and \(O(K/N)\) when the operator is evenly supported (Definition~\ref{def:even-occupancy}).
\end{theorem}

Since every anticommuting monomial lies in a selected list, the anticommuting set itself has at most \(|\qs{\gamma}|\imax\) monomials, \(O(K/N)\) for a local gate on an evenly supported operator.
The lookup cost therefore matches the size of the anticommuting set, while a dense scan costs \(O(KN)\) however few monomials anticommute.
The symmetric difference reads every entry of the selected lists, including commuting monomials of even overlap, such as \(\nu_3\) in Fig.~\ref{fig:inverted-index}, which is read twice and cancels.
The cost of the lookup is therefore set by how long the read lists are, and this is where the \(K/N\) scaling comes from.
A degree-\(d\) monomial appears in at most \(2d\) lists, so the \(2N\) lists hold at most \(2dK\) entries, on average at most \(dK/N\) each (Lemma~\ref{lem:occupancy}).
A gate reads \(|\qs{\gamma}|\) lists, so if they have similar lengths it reads \(O(|\qs{\gamma}|\,dK/N)\) entries.
At fixed \(K\), a dense scan grows as \(N\), while this shrinks as \(1/N\) (Fig.~\ref{fig:kn-schematic}).
Section~\ref{sec:single-core} demonstrates this in practice (Fig.~\ref{fig:inverse-scaling}(c)).
The average only bounds the longest list \(\imax\) when no list is much longer than the others, which we define here:
\begin{definition}[Evenly supported operator]\label{def:even-occupancy}
  An operator is \emph{evenly supported} when its longest list is within a constant factor \(c\) of the mean,
  \[
    \imax\le\frac{c}{2N}\sum_{j=1}^{2N}|\ilist[j]|.
  \]
\end{definition}
For example, if the monomials are picked at random, each position is about as likely to be occupied as any other, so the lists all end up with roughly the same length, as in Fig.~\ref{fig:kn-schematic}.
On the other hand, if every monomial contained \(Z_1\), the list for that position would hold all \(K\) monomials, and a gate reading that list, such as an \(X_1\) rotation, would read them all.

Both sources of sparsity the index takes advantage of, local gates (\(|\qs{\gamma}|=O(1)\)) and low-degree monomials (at most \(2d\) lists each), are basis-dependent, which is why \code{monoprop} propagates problems in their natural basis.
For standard fermion--qubit mappings, Pauli and Majorana monomials are related by linear maps over the binary field~\cite{chiew_ternary_2024}, so a monomial's degree in the other basis can be quickly computed when needed for truncation.
Converting, however, forfeits the \(O(K/N)\) scaling, because fermion-to-qubit encodings generally map local Majorana operators to Pauli strings whose weight grows with \(N\).
Under Jordan--Wigner~\cite{jordan_wigner}, a local Majorana generator acquires a parity string of weight \(O(N)\), and ternary-tree mappings~\cite{jiang_ternary_2020, miller_bonsai_2023, miller_treespilation_2026} reach only \(\Theta(\log N)\) average Pauli weight per Majorana.
Conversely, a single-qubit Pauli operator maps to a Majorana monomial of correspondingly large degree.

\subsection{Single-core scaling}\label{sec:single-core}

We compare \code{monoprop} with \code{PauliPropagation.jl} and \code{ppvm} (Appendix~\ref{app:versions}) on a trotterised kicked-Ising chain, and with \code{MajoranaPropagation.jl} on a trotterised two-dimensional Hubbard model in the Majorana basis.
In both cases, the observable is a magnetisation summed over all sites.
We report the time per gate per monomial, the run time divided by \(\sum_g K_g\), where \(K_g\) is the number of monomials the operator holds when gate \(g\) is applied (Appendix~\ref{app:held}).
This measures the work per gate across the evolution.

Figure~\ref{fig:inverse-scaling}(c) shows why \code{monoprop} gets faster per monomial as the system grows.
A dense code checks commutation of every monomial at each gate.
\code{monoprop} works almost only on the anticommuting set, which shrinks roughly as \(1/N\): for ten kicked-Ising layers, a gate anticommutes with \(15\%\) of the tracked monomials at \(N=32\) but only \(0.46\%\) at \(N=1024\).
To find the set, the index reads commuting monomials that overlap the gate an even number of times.
In these examples, this at most doubles the work, as the small gap between the \code{monoprop} lines and the anticommuting set line shows.

This shrinking set is what sets the time per monomial in Fig.~\ref{fig:inverse-scaling}(a) and (b).
On the kicked-Ising chain it falls at every depth as \(N^{-0.92}\) to \(N^{-0.94}\) for \code{monoprop}, close to the \(1/N\) of Theorem~\ref{thm:xor}, while it grows as \(N^{0.35}\) to \(N^{0.57}\) for \code{PauliPropagation.jl} and \(N^{0.05}\) to \(N^{0.31}\) for \code{ppvm}, each fitted over the range the codes complete.
On the Hubbard lattice it falls as \(N^{-0.83}\) to \(N^{-0.89}\) for \code{monoprop} and grows as \(N^{0.40}\) to \(N^{0.48}\) for \code{MajoranaPropagation.jl}, so the inverse scaling holds in both bases.
The ratio of dense to \code{monoprop} cost thus grows as \(N^{1.0}\) to \(N^{1.5}\), short of the asymptotic \(N^2\), because over this range the dense codes' \(O(N)\) test does not yet dominate their runtime.
\code{PauliPropagation.jl} is faster here up to \(N=64\) as for smaller system sizes the dense sorted vector implementation outperforms the sparse algebra and hashmap based implementation of \code{monoprop}.
At \(N=512\), \code{monoprop} is \(24\times\) faster than \code{PauliPropagation.jl} and \(54\times\) faster than \code{ppvm} after seven layers, and \(149\times\) faster than \code{MajoranaPropagation.jl} after three Hubbard layers.
It propagates \(3.7\times10^{8}\) monomials through thirteen kicked-Ising layers at \(N=1024\) in 868~s and \(2.0\times10^{8}\) through five Hubbard layers at \(N=968\) in 271~s, while at these sizes the reference codes exceed 30 minutes from seven and three layers, respectively.

\subsection{Distributed scaling}\label{sec:distributed-scaling}
Figure~\ref{fig:mpi-scaling} shows distributed propagation from 128 to 8192 cores on the x86 partition of the Deucalion supercomputer.
The benchmark is an open-boundary Hubbard chain of 60 lattice sites at half filling, with the spin-up occupation of lattice site 46 as the observable and cutoff \(d=10\).
We vary the number of tracked monomials through the coefficient threshold (Section~\ref{sec:background}).

In \emph{strong scaling} the problem size is fixed and the ideal wall-clock time \(\tau\) falls as \(1/C\) with core count \(C\).
In \emph{weak scaling} the problem grows with \(C\) and the ideal \(\tau\) is constant.
\emph{Parallel efficiency} is the fraction of this ideal achieved, \(\tau(C_0)C_0/[\tau(C)C]\) for strong scaling, normalised to the smallest run of each workload: 128 cores for \(1.6\times10^{9}\) monomials, 256 for \(6.1\times10^{9}\), and 1024 for \(2.44\times10^{10}\).
At 8192 cores (64 nodes), the three workloads retain 48\%, 81\%, and 99\% efficiency.
For weak scaling, relative to 128 cores, the workloads of \(9.6\times10^{7}\), \(3.85\times10^{8}\), and \(1.48\times10^{9}\) monomials per node retain 61\%, 81\%, and 94\% efficiency at 8192 cores.
The largest run tracks \(1.48\times10^{9}\) monomials per node, \(9.47\times10^{10}\) in total.

\begin{figure}[t]
  \centering
  \includegraphics[width=.85\linewidth]{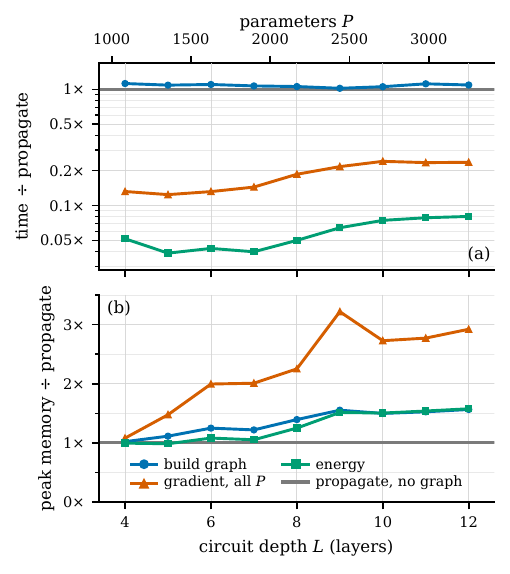}
  \caption{\label{fig:graph-gradient} \textbf{Building and evaluating the surrogate graph.}
    Single-core time and memory cost of building the surrogate graph with a single gradient and energy evaluation, relative to propagation of the same circuit with no graph.
    The circuit is a 127-qubit heavy-hex kicked-Ising model (Appendix~\ref{app:heavy-hex}), with Pauli-weight cutoff 6 and no coefficient truncation.
    Panel (a) shows the wall-clock time of one call and (b) the peak memory above the baseline.}
\end{figure}

\subsection{Surrogate graphs and gradients}\label{sec:graph-gradient}

Propagating the circuit again at each new set of angles repeats work that does not depend on the angles: finding the anticommuting monomials, and locating and inserting their partners in the hash table.
When monomials are truncated by degree only, the monomials and the pairs each gate rotates together are the same at every angle, and only the coefficients change.
A \emph{surrogate} exploits this by recording the propagation once and re-evaluating it at new angles.
Surrogates were introduced with the low-weight efficient simulation algorithm (LOWESA)~\cite{lowesa_2023,fontana_noisy_2025}, and have been implemented and studied in the context of Pauli~\cite{rudolph_pauli_propagation_2025, bermejo_qcnn_2026, lerch_patches_2026, monaco_symbolic_2025}, and Majorana propagation~\cite{majorana_propagation_2025, adapt-vmpe}.
Gradients can also be obtained without a stored surrogate, by back-propagating through the circuit~\cite{lin_backpropagating_2026}.
For each gate, \code{monoprop} stores only the coefficient-vector positions of the pairs the gate rotates together and the sign of each rotation.
Monomials that only receive a \(\cos\theta\) factor are found through the anticommutation index.
Replaying the graph at new angles applies each gate as a layer of \(2\times2\) rotations on the coefficient vector, with no lookups or insertions, and the same layers give the gradient with respect to every parameter in one forward and one backward sweep (Appendix~\ref{app:surrogate}).

Figure~\ref{fig:graph-gradient} shows the cost of each step relative to plain propagation of a trotterised 127-qubit heavy-hexagonal kicked-Ising circuit~\cite{kim_evidence_2023} up to 12 layers.
Building the graph costs about the same as propagation, so recording is cheap.
When built, evaluating the expectation value at new angles takes \(\sim0.1\times\) the propagation time, and the gradient with respect to every parameter \(\sim0.2\times\).
Building the graph keeps peak memory within \(1.0\)--\(1.6\times\) that of propagation.
The gradient evaluation needs more, reaching \(2.9\times\) at 12 layers, as its backward sweep uses two coefficient vectors.

\subsection{Optimal gates and angles}\label{sec:selection}

Adaptive variational quantum algorithms for quantum chemistry build circuits gate by gate, adding at each step the candidate from a predefined \emph{pool} that most lowers the energy, ranked by the energy gradient at zero angle or by the energy decrease at the optimal angle~\cite{grimsley_adapt_2019,adapt-vmpe,feniou_greedy_2025}.
The original pool of single and double excitations~\cite{grimsley_adapt_2019} consists of combinations of degree-2 and degree-4 Majorana monomials.
We consider rotations of Majorana operators of degree up to \(\gdeg\)~\cite{miller_treespilation_2026,adapt-vmpe}.
A pool of degree-\(\gdeg\) generators holds \(O(N^\gdeg)\) candidates, and scoring each against the \(K=O(N^d)\) terms of a degree-\(d\) observable costs \(O(N^{d+\gdeg})\).
For Fock reference states we halve this exponent.
To do this, we notice a Fock state is a combination of \emph{diagonal} monomials, products of mode parities \(\iu m_{2i-1}m_{2i}\), so only pairs of candidate and Hamiltonian monomials whose product is diagonal contribute, and we enumerate those pairs directly (Section~\ref{sec:selection-methods}).\footnote{The methods described here are not yet available in our open source release of \code{monoprop}.}

\begin{figure}[tb]
  \centering
  \includegraphics[width=0.85\linewidth]{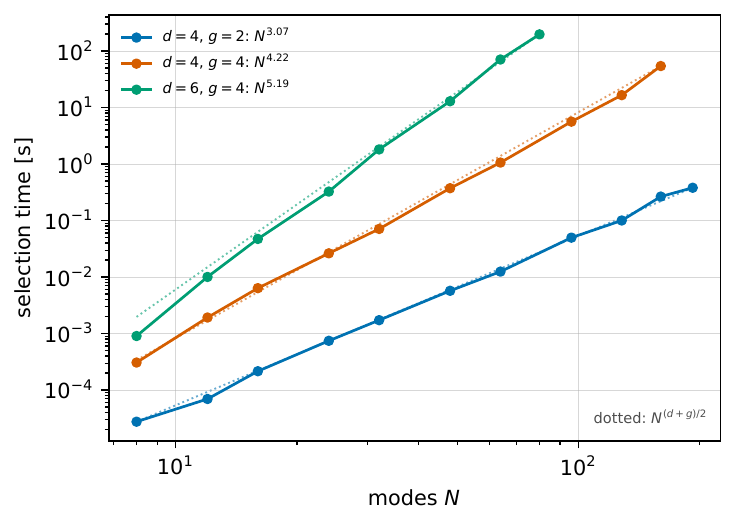}
  \caption{\label{fig:selection-scaling} \textbf{Optimal gate and angle generation.}
    Single-core timings for generation and scoring of generated candidate gates, with \((d,\gdeg)\in\{(4,2),(4,4),(6,4)\}\) for a chemistry Hamiltonian and Hartree--Fock reference.
    Legend exponents are fitted over the four largest system sizes.
    Dotted guides show \(N^{(d+\gdeg)/2}\).}
\end{figure}

\begin{theorem}[Optimal gates and angles]\label{thm:energy-selection}
  Let \(H\) be an even-parity Majorana operator, a sum of monomials of degree at most \(d\), and \(\varrho = \ketbra{n}{n}\) a Fock state on \(N\) modes.
  For even \(2\le\gdeg\le d\), one pass computes the energy gradient, optimal angle, and energy decrease of every gate whose generator places one Majorana on each of at most \(\gdeg\) distinct modes, in \(O(N^{(d+\gdeg)/2})\) time.
\end{theorem}
Restricting to one Majorana per mode loses no candidates as a gate containing both Majoranas \(m_{2i-1}m_{2i}\) of a mode acts on a Fock state, up to a sign, as the lower-degree generator with the pair removed (Appendix~\ref{app:gate-selection}).
The resulting ranking is exact up to truncation errors in the evolved Hamiltonian.
The measured exponents \(3.1\), \(4.2\), and \(5.2\) match the predicted \((d+\gdeg)/2=3\), \(4\), and \(5\) (Fig.~\ref{fig:selection-scaling}).

\section{Methods}\label{sec:methods}

\subsection{Gate update}\label{sec:gate-update}

The Pauli and Majorana simulators share data structures and the propagation algorithm.
The basis enters only through the phase of a product (Appendix~\ref{app:multiplicative-phase}), the degree, and the anticommutation support \(\qs{\gamma}\).
We represent the operator as a vector of sparse monomials and coefficients, a hash table, which finds a monomial's position in these vectors in expected constant time, and an anticommutation index.
To save memory, each of the \(2N\) lists is stored either as a sparse list of entries or as a dense bitmap, whichever is smaller, similar to Roaring bitmaps~\cite{chambi_roaring_2016}.

We apply \(U = e^{-\iu\theta M_\gamma/2}\) in three steps:
\begin{enumerate}
  \item \textbf{Lookup.} Find \(\anti(\gamma)\) from the index (Theorem~\ref{thm:xor}) and scale by \(\cos\theta\).
  \item \textbf{Split.} A monomial \(\nu\) and its partner \(\nu\symd\gamma\) both anticommute with the generator, so \(\anti(\gamma)\) consists of pairs that must each be rotated once.
        We split each pair into a \emph{leader}, which omits a fixed position \(j^\star\in\gamma\), and a \emph{follower}, which contains it (Appendix~\ref{app:leader-follower}).
  \item \textbf{Rotate or branch.} For each leader, then each unmarked follower, one hash-table lookup finds its partner.
        If the partner is tracked, the two are rotated together and the follower is marked.
        Otherwise, the partner is added unless it is truncated.
\end{enumerate}
Each pair of anticommuting monomials costs one hash-table lookup and \(O(d+|\gamma|)\) further work, so when \(d, |\gamma| \in O(1)\) the cost per gate is
\begin{equation}
  \label{eq:time-per-gate}
  O\bigl(|\qs{\gamma}|\imax + |\anti(\gamma)|\bigr) = O(\imax).
\end{equation}
Table~\ref{tab:gate-cost} compares each step with a dense bitset.
On an evenly supported operator the lookup falls from \(O(KN)\) to \(O(|\qs{\gamma}|\,dK/N)\), rotating and merging each monomial costs \(O(d+|\gamma|)\) rather than \(O(N)\), and the memory per monomial falls from \(O(N)\) to \(O(d)\).

\begin{table}[t]
  \centering
  \footnotesize
  \setlength{\tabcolsep}{4pt}\renewcommand{\arraystretch}{1.35}
  \begin{tabular}{@{}>{\raggedright\arraybackslash}p{0.39\linewidth}ll@{}}
    \toprule
    Step                                                                            & Dense bitset  & \code{monoprop}                     \\
    \midrule
    Find the \(M_\nu\) anticommuting with \(M_\gamma\)                              & \(O(KN)\)     & \(O(|\qs{\gamma}|\,dK/N)\)          \\
    Rotate each, \(M_\nu\mapsto\cos\theta\,M_\nu\pm\sin\theta\,M_{\nu\symd\gamma}\) & \(O(N)\) each & \(O(d+|\gamma|)\) each              \\
    Merge \(M_{\nu\symd\gamma}\) into \(H\), truncating at degree \(d\)             & \(O(N)\) each & \(O(d+|\gamma|)\) each              \\
    \midrule
    Per gate                                                                        & \(O(KN)\)     & \(O(|\gamma|\,d\,(d+|\gamma|)K/N)\) \\
    Memory per monomial                                                             & \(O(N)\)      & \(O(d)\)                            \\
    \bottomrule
  \end{tabular}
  \caption{\label{tab:gate-cost} \textbf{Cost of propagating a gate} as \(H\mapsto e^{\iu\theta M_\gamma/2}  H e^{-\iu\theta M_\gamma/2}\), for \(K\) monomials of degree at most \(d\) on \(N\) sites.
    Here \(M_\gamma\) is the gate generator, \(|\gamma|\) the number of positions it occupies, and \(\nu\symd\gamma\) the support of the product \(M_\nu M_\gamma\).
    \code{monoprop} costs assume an evenly supported operator (Definition~\ref{def:even-occupancy}).}
\end{table}

\subsection{Parallelisation and distribution}\label{sec:distributed}

\code{monoprop} runs as several MPI processes, called \emph{ranks}, which may exist on different nodes and exchange data by messages.
Each rank runs several threads in shared memory, where the operator is divided among these threads into disjoint \emph{partitions}, each holding the data structures of Section~\ref{sec:gate-update} for its monomials.
A monomial \(\nu\) is owned by rank \(\rank(\nu)\) among the \(R\) ranks.
When \(R\) is a power of two, \(\rank(\nu)\) is the bitwise XOR (\(\oplus\)) of fixed labels assigned to the positions in \(\nu\).
This map is linear, \(\rank(\nu\symd\gamma)=\rank(\nu)\oplus\rank(\gamma)\), so under the gate generated by \(\gamma\) every monomial on rank \(r\) finds its partner on the single rank \(r\oplus\rank(\gamma)\) (Appendix~\ref{app:rank-routing}), similar to the bit-wise partitioning of parallel iterative qubit coupled cluster (iQCC) simulations~\cite{hosseini_parallel_iqcc_2026}.
Within a rank, a generic hash function spreads monomials evenly among partitions.
In the distributed benchmarks, each 128-core node runs 8 ranks of 16 cores with one partition per core, and each rank is pinned to one NUMA domain, so memory accesses stay local and fast.

\subsection{Gate selection}\label{sec:selection-methods}

Theorem~\ref{thm:energy-selection} rests on two ideas.
First, the energy under a candidate gate \(U=e^{-\iu\theta M_\gamma/2}\) is~\cite{nakanishi_sequential_2020,ostaszewski_structure_2021},
\[
  f_\gamma(\theta)=\Tr[\varrho\,U^\dagger H U]=A_\gamma\cos\theta+B_\gamma\sin\theta+C_\gamma,
\]
so the optimal angle and energy decrease follow from \(A_\gamma\) and \(B_\gamma\).
\(A_\gamma\) is half the energy change on flipping the occupations of the modes \(\gamma\) acts on, so it depends only on the diagonal part of \(H\).
Second, rather than checking the \(K\) monomials for contributions to \(B_\gamma\), we generate them from the reference state.
A Fock state is a sum of diagonal monomials, the only monomials with non-zero Fock-state expectation (Appendix~\ref{app:reference-states}), so a monomial \(M_\nu\) contributes to \(B_\gamma\) only when \(M_\gamma M_\nu\) is diagonal.
We therefore generate the diagonal products and the pairs \((\gamma,\nu)\) that can produce them directly.
Every mode such a pair touches carries exactly two of its at most \(d+\gdeg\) Majoranas, so the pair spans at most \((d+\gdeg)/2\) modes, which gives the bound.
Appendix~\ref{app:gate-selection} gives the derivation, algorithm, proof and additional details.

\section{Discussion}\label{sec:discussion}

The anticommutation index replaces the commutation check of every tracked monomial with a lookup in a few lists, lowering the per-gate cost for local gates on evenly supported operators from \(O(KN)\) to \(O(K/N)\), an asymptotically quadratic improvement in \(N\).
The same data structures distribute across thousands of cores and underlie the surrogate graph and gate selection.

The \(O(K/N)\) bound we observe numerically needs gates and monomials that occupy \(O(1)\) positions in the indexed basis, and an operator that stays evenly supported, which we do not prove is always the case.

How evenly the monomials spread depends on the observable and the dynamics.
An observable summed over the whole system, such as the magnetisation of Fig.~\ref{fig:inverse-scaling}, covers it from the start, while the monomials of a local observable stay within its light cone, and gates outside it read empty lists.
How the dynamics, for example how strongly they scramble, affect the list lengths is left to future work.
Even in the worst case, where the operator is concentrated on the few positions a gate reads, the lookup costs \(O(K)\) (Theorem~\ref{thm:xor}).
This is still a factor of \(N\) below the \(O(KN)\) of a dense scan, and in reality the cost may be amortised over the other gates in the circuit that act on emptier lists.

We have focused on performance and scalability, not on the accuracy of observable estimates, variational optimisation, state preparation, or integration with quantum algorithms.
The falling per-monomial cost makes it possible to tackle these problems at large system sizes.
Future work on \code{monoprop} will reduce the memory per monomial, for instance by replacing the hash table, which would also avoid its slow, scattered memory accesses.
Beyond propagation, these data structures for monomials may serve other latency-sensitive tasks, such as stabiliser simulation~\cite{aaronson_gottesman_2004}, real-time quantum error-correction decoding~\cite{battistel_realtime_2023}, or hybrid quantum--classical workloads.

\subsection*{Acknowledgements}
The author thanks Roberto Di Remigio Eikås, Matteo Rossi, and the Algorithmiq software development team for valuable discussions and their effort in open-sourcing \code{monoprop}.
The author thanks Tommaso Tufarelli, Adam Glos, Guillermo García-Pérez, Manuel Rudolph and Zoë Holmes for valuable scientific discussions and comments on the paper.
The author thanks Guillermo García-Pérez for pointing out the monomial representation of Pauli strings.
The author acknowledges the EuroHPC Joint Undertaking for awarding access to Deucalion, hosted by FCCN, Portugal (EHPC-DEV-2026D08-260) and to Leonardo, hosted by CINECA, Italy (EHPC-DEV-2026D09-033).

\subsection*{AI Usage Statement}
The author used Claude to improve clarity and to simplify long-winded proofs and for presentation.
Data structures, algorithms, and proofs were conceived by the author.
Claude also helped implement the code and benchmarks.
AI was not used to generate scientific content, references, or conclusions, and the author checked all content for accuracy.

\subsection*{Competing Interests}
Elements of this work are included in patents filed by
Algorithmiq Ltd with the European and US Patent Offices.

\bibliographystyle{unsrt}

\clearpage
\onecolumn
\appendix

\thispagestyle{empty}
\vspace*{\fill}
\begin{center}
  {\Large\bfseries Supplementary Material for \code{monoprop}: Scalable Pauli and Majorana propagation via anticommutation indices\par}
  \vspace{2em}
  {\normalsize Aaron Miller\par}
  \vspace{0.6em}
  {\small
    \textit{Algorithmiq S.r.l., Via della Chiusa 15, 20123 Milan, Italy}\par
    \textit{School of Physics, Trinity College Dublin, College Green, Dublin 2, Ireland}\par}
\end{center}
\vspace*{\fill}
\startcontents[supplementary]
\section*{Contents}
\printcontents[supplementary]{}{1}{\setcounter{tocdepth}{2}}
\clearpage

\section{Notation and conventions}\label{app:notation}

The following conventions apply to Pauli and Majorana bases, with \([k]\coloneq\{1,\ldots,k\}\).

\begin{itemize}
  \item \textbf{Sites and positions.} A register has \(N\) physical sites: fermionic modes in the Majorana basis or qubits in the Pauli basis.
        Each site \(i\) has two positions, \(2i-1\) and \(2i\), labelling the operators \(m_{2i-1},m_{2i}\) or \(X_i,Z_i\), respectively.
        Positions therefore label operator factors, not additional physical sites, and a monomial may occupy neither, one, or both positions of a site.
        In the Pauli basis these four choices represent \(I_i,X_i,Z_i,Y_i\).
        We use \(j\) for a position and \(\ell\) for a gate or surrogate layer.
  \item \textbf{Monomials.} A monomial \(M_\nu\) is an ordered product of the operators selected by \(\nu\subseteq[2N]\), with a fixed phase making it Hermitian (Appendix~\ref{app:monomial-algebra}).
        Its support \(\supp(M_\nu)=\nu\) records the selected positions, while \(\sites(\nu)\subseteq[N]\) records the sites containing them.
        The \(4^N\) such monomials form an operator basis, so an observable is expanded as \(H=\sum_\nu c_\nu M_\nu\), and the empty support gives \(M_\varnothing=I\).
        As in Section~\ref{sec:background}, \(\nu\) also denotes the monomial.
  \item \textbf{Symmetric difference.} We write \(\symd\) for symmetric difference of occupied-position sets and of monomial sets such as index lists.
        We reserve \(\oplus\) for bitwise XOR of rank labels.
  \item \textbf{Degree and cutoff.} The degree \(\dg{\nu}\) is Majorana length \(|\nu|\) or Pauli weight \(|\sites(\nu)|\).
        Cardinality counts positions, so in the Pauli basis \(\dg{\nu}\le|\nu|\le2\dg{\nu}\).
        The tracked family \(\mathcal K_d\) contains \(K\) monomials of degree at most \(d\). The generator degree \(\gdeg\) (Appendix~\ref{app:gate-selection}) is distinct from this cutoff.
\end{itemize}

\section{Majorana and Pauli monomials}\label{app:monomial-algebra}
\paragraph{Majorana monomials}
As in Eq.~\eqref{eq:majorana-monomial}, a Hermitian Majorana monomial is
\begin{equation}
  M_\nu = \iu^{\binom{|\nu|}{2}}\, m_{j_1} m_{j_2} \cdots m_{j_{|\nu|}},
  \qquad j_1 < j_2 < \cdots < j_{|\nu|},
\end{equation}
where the phase makes the ordered product Hermitian.
Two Majorana monomials satisfy
\begin{equation}
  M_\nu M_\mu = (-1)^{|\nu||\mu| + |\nu \cap \mu|} M_\mu M_\nu ,
  \label{eq:majorana-commutation}
\end{equation}
so, when at least one of the two lengths is even, they anticommute exactly when the overlap \(|\nu \cap \mu|\) is odd.
We usually consider even-parity fermionic systems, such as the Hubbard model (Appendix~\ref{app:hubbard}), in which the observable and the generators contain only even-length monomials.

\paragraph{Pauli monomials}
With \(X(\nu)\), \(Z(\nu)\), and \(y(\nu) = |X(\nu) \cap Z(\nu)|\) as in Eq.~\eqref{eq:pauli-monomial}, a Hermitian Pauli monomial is
\begin{equation}
  M_\nu = \iu^{y(\nu)} \prod_{i \in X(\nu)} X_i \prod_{i \in Z(\nu)} Z_i,
\end{equation}
with all \(X\) factors written before all \(Z\) factors.
The prefactor turns each doubly occupied site into \(Y_i = \iu X_i Z_i\), so the weight \(\dg{\nu} = |X(\nu) \cup Z(\nu)|\) satisfies \(|\nu| = \dg{\nu} + y(\nu)\).
Writing \(\mu^{\times}\) for \(\mu\) with the two positions within every site swapped, two Pauli monomials satisfy
\begin{equation}
  M_\nu M_\mu = (-1)^{|\nu \cap \mu^{\times}|} M_\mu M_\nu ,
  \qquad
  |\nu \cap \mu^{\times}| = |Z(\nu) \cap X(\mu)| + |X(\nu) \cap Z(\mu)| ,
  \label{eq:pauli-commutation}
\end{equation}
so they anticommute exactly when the exchanged overlap is odd.

\paragraph{Trace and orthogonality}
In either basis, taking the trace over the full \(2^N\)-dimensional Hilbert space gives
\begin{equation}
  \Tr[M_\nu] = 2^N\delta_{\nu,\varnothing},
  \qquad
  \Tr[M_\nu M_\mu] = 2^N\delta_{\nu,\mu},
  \label{eq:monomial-trace}
\end{equation}
where \(\delta\) is the Kronecker delta on monomial labels.

\subsection{Monomial multiplication}\label{app:multiplicative-phase}

A branching gate multiplies a tracked monomial by its generator.
In both bases the product of two monomials is, up to a phase, the monomial labelled by their symmetric difference,
\begin{equation}
  M_\nu M_\mu = \chi(\nu, \mu)\, M_{\nu\symd\mu},
  \qquad
  \chi(\nu, \mu) = \chi_{\mathrm h}(\nu,\mu)\, \chi_{\mathrm o}(\nu,\mu) .
  \label{eq:monomial-product}
\end{equation}
The phase factorises into the ordering sign \(\chi_{\mathrm o}(\nu,\mu) \in \{\pm 1\}\), from restoring the canonical operator order, and the phase \(\chi_{\mathrm h}(\nu,\mu)\), from the Hermitian prefactors of the three monomials involved.

\paragraph{Majorana phase}
The prefactors of \(M_\nu\), \(M_\mu\), and \(M_{\nu\symd\mu}\) give
\begin{equation}
  \chi_{\mathrm h}(\nu,\mu) = \iu^{\binom{|\nu|}{2} + \binom{|\mu|}{2} - \binom{|\nu\symd\mu|}{2}} \in \{\pm 1, \pm \iu\} .
\end{equation}
In the product of the two ordered Majorana strings, each Majorana of the second factor must be moved left past every Majorana of the first that stands to its right, and each such swap contributes a minus sign because
\begin{equation}
  m_j m_{j'} = - m_{j'} m_j, \qquad j \neq j' .
\end{equation}
The total number of crossings is therefore the number of \emph{inversions} of the word formed by concatenating the ordered positions of \(\nu\) and \(\mu\):
\begin{equation}
  \invcount(\nu,\mu) = \bigl|\{\,(j,j') \in \nu \times \mu \;:\; j > j' \,\}\bigr| ,
  \label{eq:crossing-count}
\end{equation}
of which only the parity matters:
\begin{equation}
  \chi_{\mathrm o}(\nu,\mu) = (-1)^{\invcount(\nu,\mu)} .
  \label{eq:reordering-sign}
\end{equation}
A position shared by the two monomials contributes \(m_j^2 = I\).

\paragraph{Pauli phase}
The prefactors of \(M_\nu\), \(M_\mu\), and \(M_{\nu\symd\mu}\) give
\begin{equation}
  \chi_{\mathrm h}(\nu,\mu) = \iu^{\,y(\nu) + y(\mu) - y(\nu\symd\mu)} \in \{\pm 1, \pm \iu\} .
\end{equation}
In the product of the two \(X\)-then-\(Z\) strings, restoring this order moves the \(X\) factors of \(\mu\) past the \(Z\) factors of \(\nu\).
Factors on different sites commute, and each site in \(Z(\nu)\cap X(\mu)\) contributes \(ZX=-XZ\), so
\begin{equation}
  \chi_{\mathrm o}(\nu,\mu) = (-1)^{|Z(\nu)\,\cap\,X(\mu)|} .
  \label{eq:pauli-sign}
\end{equation}

\subsection{Reference states}\label{app:reference-states}

We contract the back-propagated observable against a computational-basis state (Pauli basis) or a Fock state (Majorana basis), with occupations \(n_i\in\{0,1\}\).
Both reference states have the form
\begin{equation}
  \varrho = \ketbra{n_1\ldots n_N}
  = \frac{1}{2^N}\prod_{i=1}^{N}\bigl(I + \zeta_i D_i\bigr),
  \label{eq:product-reference}
\end{equation}
where the local operators \(D_i\) commute, square to the identity, and have reference-state eigenvalues \(\zeta_i\).

\paragraph{Majorana basis}
Here \(D_i=\iu\,m_{2i-1}m_{2i}\) and \(\zeta_i=(-1)^{n_i+1}\).
The \emph{diagonal monomials} occupy either both positions of each mode or neither.
A monomial on \(k\) modes has length \(2k\), giving \(\binom{N}{d/2}\) possible diagonal monomials of exact degree \(d\) for even \(d\), and none for odd \(d\).

\paragraph{Pauli basis}
Here \(D_i=Z_i\) and \(\zeta_i=(-1)^{n_i}\).
Here the diagonal monomials contain only \(I\) and \(Z\).
One on \(k\) sites has weight \(k\), giving \(\binom{N}{d}\) possible diagonal monomials of exact degree \(d\).

\paragraph{Contraction}
In both bases, expanding the product gives
\begin{equation}
  \varrho = \frac{1}{2^N}\sum_{S \subseteq [N]} \Bigl(\prod_{i \in S} \zeta_i\Bigr) \prod_{i \in S} D_i .
  \label{eq:product-reference-expansion}
\end{equation}
Thus \(\varrho=\sum_\nu w_\nu M_\nu\) has coefficient \(w_\nu=2^{-N}\prod_{i\in\sites(\nu)}\zeta_i\) on diagonal monomials and zero otherwise.
For \(H_L=\sum_\nu c_\nu M_\nu\), trace orthogonality (Eq.~\eqref{eq:monomial-trace}) gives
\begin{equation}
  \langle H_L \rangle
  = \Tr[\varrho H_L]
  = 2^N \sum_{\nu} c_\nu w_\nu
  = \sum_{\nu \text{ diagonal}} c_\nu \prod_{i \in \sites(\nu)} \zeta_i ,
\end{equation}
where the final sum runs only over tracked diagonal monomials.
The expectation value therefore needs only the tracked diagonal monomials, and the reference-state expansion is never stored.
This implicit treatment applies to Heisenberg-picture propagation.
In the Schrödinger picture, which \code{monoprop} also supports, the degree-truncated reference state is evolved explicitly.

\section{The anticommutation index}\label{app:anticommutation-index}

This appendix proves Theorem~\ref{thm:xor} for the index of Definition~\ref{def:inverted-index}.
Let the tracked monomials \(\mathcal{K}_d\) obey a fixed degree cutoff \(d\ge1\), so \(\dg{\nu} \le d\) for every \(\nu \in \mathcal{K}_d\) and \(K = |\mathcal K_d| = O(N^d)\).
For a generator \(\gamma\), write
\[
  \anti(\gamma)
  \coloneq
  \bigl\{\,\nu \in \mathcal{K}_d : \{M_\nu, M_\gamma\} = 0 \,\bigr\}
\]
for the \emph{anticommuting set} of tracked monomials.
The generator's anticommutation support \(\qs{\gamma}\) is defined by the parity test of Eq.~\eqref{eq:parity-test}, with basis-specific forms in Appendix~\ref{app:query-support}.

\subsection{Occupancy bounds}\label{app:occupancy}

\begin{lemma}[Occupancy bounds]\label{lem:occupancy}
  Let \(\mathcal{K}_d\) contain \(K\) monomials of degree at most \(d\).
  Then
  \begin{equation}
    \sum_{j=1}^{2N}|\ilist[j]|\le 2dK,
    \qquad
    \frac{1}{2N}\sum_{j=1}^{2N}|\ilist[j]|\le\frac{dK}{N}.
    \label{eq:total-occupancy}
  \end{equation}
  Thus the index has \(O(K)\) entries at fixed \(d\), and the mean list occupancy is at most \(dK/N\).
  The peak occupancy obeys
  \begin{equation}
    \frac{K-1}{2N} \;\le\; \imax \;\le\; \min\bigl\{K,\ O(N^{d-1})\bigr\} .
    \label{eq:occupancy-sandwich}
  \end{equation}
\end{lemma}
\begin{proof}
  Each monomial \(\nu\) contributes exactly \(|\nu|\) entries, so
  \[
    \sum_{j=1}^{2N}|\ilist[j]|=\sum_{\nu\in\mathcal K_d}|\nu|.
  \]
  Each monomial occupies at most \(2d\) positions, so the right-hand side is at most \(2dK\), and dividing by the \(2N\) positions gives Eq.~\eqref{eq:total-occupancy}.
  Every monomial apart from the identity contributes at least one entry, so the total is at least \(K-1\) and the longest list satisfies \((K-1)/(2N) \leq \imax\).

  For the upper bound, a list holds at most \(K\) monomials, with equality when every tracked monomial occupies that position.
  Now consider the family \(\mathcal K_{\mathrm{full}}\) of all monomials of degree at most \(d\), with \(K_{\mathrm{full}}=|\mathcal K_{\mathrm{full}}|=O(N^d)\).
  This family is invariant under permutations of the sites and under exchanging the two positions of any one site, since both preserve Majorana length and Pauli weight.
  These symmetries map any position to any other, so all lists of the full basis have equal occupancy, which by the bound above is at most \(dK_{\mathrm{full}}/N\).
  Since \(\mathcal K_d\subseteq\mathcal K_{\mathrm{full}}\), its lists are no longer than those of the full basis, so for every \(j\),
  \[
    |\ilist[j]|\le\frac{dK_{\mathrm{full}}}{N}=O(N^{d-1}).
  \]
  Combining the two bounds gives \(\imax\le\min\{K,O(N^{d-1})\}\).
\end{proof}
Spreading \(O(dK)\) entries over \(2N\) positions gives an \(O(K/N)\) mean, the smallest possible scale by the lower bound.
The family \(\mathcal K_{\mathrm{full}}\) attains it, and so does any evenly supported operator (Definition~\ref{def:even-occupancy}), for which \(\imax=\Theta(K/N)\).

\subsection{Anticommutation lookup}\label{app:query}

\begin{proof}[Proof of Theorem~\ref{thm:xor}]
  A monomial \(\nu\) appears in one selected list for every position in \(\nu\cap\qs{\gamma}\).
  The symmetric difference keeps \(\nu\) exactly when this number is odd, which is the anticommutation test of Eq.~\eqref{eq:parity-test}, so the lookup returns \(\anti(\gamma)\).
  Each selected list contains at most \(\imax\) entries, so the \(|\qs{\gamma}|\) lists selected contain at most \(|\qs{\gamma}|\imax\) entries in total, giving the stated bound.
  For \(|\qs{\gamma}|=O(1)\), \(\imax\le K\) gives the worst-case \(O(K)\) bound, while an evenly supported operator gives \(O(\imax)=O(K/N)\).
\end{proof}

The bound \(O(|\qs{\gamma}|\imax)\) is stated in the occupancy of the lists read.
Monomials appearing in an even number of selected lists cancel in the symmetric difference, so the output cardinality alone does not bound the lookup cost, even when the operator is evenly supported.
When \(\qs{\gamma}=\{j\}\), nothing cancels and the lookup enumerates exactly \(\anti(\gamma)=\ilist[j]\).
When \(|\qs{\gamma}| = \Theta(N)\), the lookup traverses at most the whole index, \(O(dK) = O(K)\) by Lemma~\ref{lem:occupancy}, and the improvement over a scan disappears.

\subsection{Anticommutation supports}\label{app:query-support}
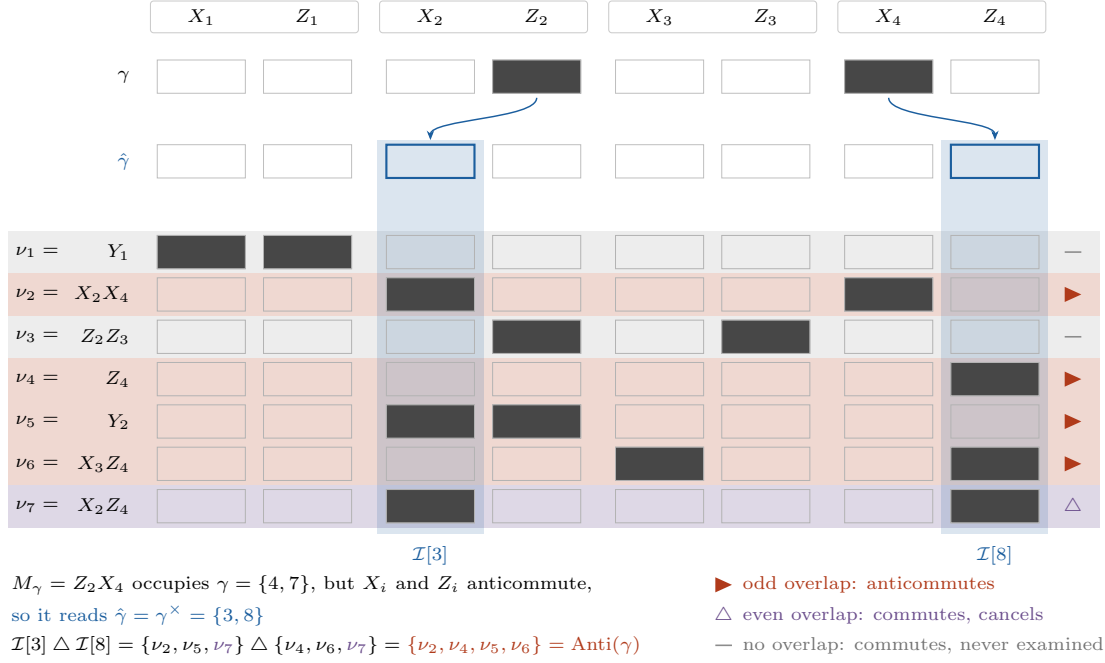
\begin{figure}[t!]
  \centering
  %% figures/inverted-index-pauli.tex --- the same index, read by the exchange.
%% Included from appendix.tex inside a figure environment; needs tikz with the
%% backgrounds and arrows.meta libraries and the notation macros \ilist, \qs,
%% \anti of the main preamble.
%%
%% The companion to the Majorana figure, and deliberately the same picture:
%% same grid width, same row pitch, same verdict column, same footer band with
%% the legend at \kx. Only one thing is added, and it is the whole content of
%% the figure --- the generator occupies \gamma but reads \gamma^{\times},
%% because the two positions of a site are X_i and Z_i, which anticommute.
%% The index itself is untouched.
%%
%% The example: N = 4 sites, so 2N = 8 positions, position 2i-1 being X_i and
%% position 2i being Z_i; cutoff d = 2; K = 7 monomials, few enough that each
%% is named and every verdict can be checked by hand. The generator is Z_2X_4,
%% occupying positions {4,7}, whose exchange is {3,8} = {X_2, Z_4}. Note that
%% \nu_4 = Z_4 anticommutes with it while sharing no position with it at all,
%% which is what the exchange is for.
%%
%% The axis carries the labels X_i, Z_i and not the position indices: X_2 and
%% Z_2 name their own site, and the exchange is read off the arrows rather than
%% off the arithmetic j -> j^{\times}. The two positions the lookup names are
%% still identified below the grid, by the \ilist[j] labels under their columns.
%% The parameters and the two complexity readings live in the caption.
%%
%% To change the example, edit the data macros below; the geometry derives from
%% \cw, \gp and \rh.

\definecolor{qcol}{HTML}{1B5E9E}   % the lookup: positions read, lists traversed
\definecolor{acol}{HTML}{B03A1A}   % the answer: the anticommuting set
\definecolor{ccol}{HTML}{6B4E8C}   % visited by the lookup, but commutes

\begin{tikzpicture}[
    x=1cm, y=1cm,
    cell/.style  = {draw=black!30, line width=0.25pt},
    onbit/.style = {draw=black!30, line width=0.25pt, fill=black!72},
    gcell/.style = {draw=black!25, line width=0.25pt},
    gbit/.style  = {draw=black!30, line width=0.25pt, fill=black!72},
    qbox/.style  = {draw=qcol, line width=0.8pt},
    hd/.style    = {font=\scriptsize, inner sep=1pt},
    note/.style  = {font=\scriptsize, align=left, inner sep=1pt},
    tile/.style  = {draw=black!20, line width=0.4pt, rounded corners=1.2pt,
        fill=white},
    xarr/.style  = {draw=qcol, line width=0.55pt,
          -{Stealth[length=3.6pt,width=2.8pt]}},
  ]

  % ---------------------------------------------------------------- geometry
  %% \cw is chosen so that the grid is as wide as the Majorana figure's, and
  %% \rh, \vx, \bx and the footer baselines are shared with it outright.
  \def\cw{1.405}                   % column pitch
  \def\gp{0.22}                    % gap between neighbouring sites
  \def\rh{0.56}                    % row pitch
  \def\gh{3.92}                    % depth of the grid: 7*\rh
  \def\gygam{2.04}                 % centre of the generator support strip
  \def\yarr{1.48}                  % height of the exchange arrows
  \def\gyqs{0.92}                  % centre of the lookup strip
  \def\lx{-1.85}                   % left edge of the row bands, and of the footer
  \def\nux{-1.06}                  % right edge of the \nu_i = column
  \def\vx{12.12}                   % the verdict glyphs
  \def\bx{12.58}                   % right edge of the row bands
  \def\cx#1{(#1-1)*\cw+int((#1-1)/2)*\gp+0.5*\cw}
  \def\cy#1{-((#1-1)*\rh+0.5*\rh)}
  \def\boxat#1#2{({\cx{#2}-0.5*\cw+0.12},{#1-0.5*\rh+0.06})
    rectangle ({\cx{#2}+0.5*\cw-0.12},{#1+0.5*\rh-0.06})}
  \def\cellat#1#2{\boxat{\cy{#1}}{#2}}

  %% The axis is one row deep, as in the Majorana figure: one tile per site,
  %% holding the two labels that site's positions carry.
  \def\ytop{3.04}                  % top of the tiles
  \def\ypau{2.83}                  % the two labels of each site
  \def\ybot{2.62}                  % foot of the tiles

  %% The footer band: three baselines under the grid, shared by the columnwise
  %% reading on the left and the legend at \kx on the right, so that the two
  %% blocks fill one band rather than hanging separately.
  \def\yla{-4.68}
  \def\ylb{-5.08}
  \def\ylc{-5.48}
  \def\kx{7.50}                    % left edge of the legend glyphs

  % -------------------------------------------------------------------- data
  %% \occ: the tracked supports. \nm: their names as Pauli monomials.
  \def\occ{1/{1,2},2/{3,7},3/{4,6},4/{8},5/{3,4},6/{5,8},7/{3,8}}
  \def\nm{1/{Y_1},2/{X_2X_4},3/{Z_2Z_3},4/{Z_4},5/{Y_2},6/{X_3Z_4},7/{X_2Z_4}}
  \def\gam{4,7}                    % \gamma, the positions Z_2X_4 occupies
  \def\qsg{3,8}                    % \qs{\gamma} = \gamma^{\times} = {X_2, Z_4}
  \def\hitrows{2,4,5,6}            % odd overlap: anticommutes, returned
  \def\evenrow{7}                  % even overlap: visited, commutes, cancels
  \def\greyrows{1,3}               % no overlap: commutes, never examined

  % ------------------------------------------- row states and columns read
  \begin{scope}[on background layer]
    \foreach \i in \greyrows
    {\fill[black!7] ({\lx},{\cy{\i}-0.5*\rh}) rectangle
        ({\bx},{\cy{\i}+0.5*\rh});}
    \foreach \i in \hitrows
    {\fill[acol!20] ({\lx},{\cy{\i}-0.5*\rh}) rectangle
        ({\bx},{\cy{\i}+0.5*\rh});}
    \foreach \i in \evenrow
    {\fill[ccol!22] ({\lx},{\cy{\i}-0.5*\rh}) rectangle
        ({\bx},{\cy{\i}+0.5*\rh});}
    \foreach \j in \qsg {
      \fill[qcol,fill opacity=0.16] ({\cx{\j}-0.5*\cw},{\gyqs+0.5*\rh})
      rectangle ({\cx{\j}+0.5*\cw},{-\gh-0.08});
    }
  \end{scope}

  % ------------------------------------------------------------- the axis
  %% The two positions of a site carry X_i and Z_i, which anticommute with one
  %% another: that single fact is what forces the exchange below.
  \foreach \s in {1,...,4} {
      \draw[tile] ({\cx{2*\s-1}-0.5*\cw+0.03},{\ybot})
      rectangle ({\cx{2*\s}+0.5*\cw-0.03},{\ytop});
      \node[hd] at ({\cx{2*\s-1}},{\ypau}) {\(X_{\s}\)};
      \node[hd] at ({\cx{2*\s}},{\ypau})   {\(Z_{\s}\)};
    }

  % ------------------------------- the generator and the lookup it must make
  \node[hd,anchor=east] at (-0.20,{\gygam}) {\(\gamma\)};
  \foreach \j in {1,...,8} {\draw[gcell] \boxat{\gygam}{\j};}
  \foreach \j in \gam     {\draw[gbit]  \boxat{\gygam}{\j};}

  \node[hd,anchor=east,text=qcol] at (-0.20,{\gyqs}) {\(\qs{\gamma}\)};
  \foreach \j in {1,...,8} {\draw[gcell] \boxat{\gyqs}{\j};}
  \foreach \j in \qsg     {\draw[qbox]  \boxat{\gyqs}{\j};}

  %% Each arrow stays inside its own site, so the exchange is visibly a swap of
  %% the two positions of one site and needs no label.
  \draw[xarr] ({\cx{4}},{\gygam-0.5*\rh}) .. controls
    ({\cx{4}},{\yarr}) and ({\cx{3}},{\yarr}) .. ({\cx{3}},{\gyqs+0.5*\rh});
  \draw[xarr] ({\cx{7}},{\gygam-0.5*\rh}) .. controls
    ({\cx{7}},{\yarr}) and ({\cx{8}},{\yarr}) .. ({\cx{8}},{\gyqs+0.5*\rh});

  % ------------------------------------------------------- operator and matrix
  %% Two columns: the \nu_i names line up on the right edge of their own column
  %% at \nux, and the monomials sit flush against the grid.
  \foreach \i/\n in \nm {
    \node[hd,anchor=east] at ({\nux},{\cy{\i}}) {\(\nu_{\i} ={}\)};
    \node[hd,anchor=east] at (-0.20,{\cy{\i}}) {\(\n\)};
  }
  \foreach \i in {1,...,7}
    {\foreach \j in {1,...,8} {\draw[cell] \cellat{\i}{\j};}}
  \foreach \i/\ps in \occ {\foreach \j in \ps {\draw[onbit] \cellat{\i}{\j};}}

  % ----------------------------------------------------------- the verdicts
  \foreach \i in \hitrows
  {\fill[acol] ({\vx},{\cy{\i}-0.10}) -- ({\vx},{\cy{\i}+0.10})
    -- ({\vx+0.22},{\cy{\i}}) -- cycle;}
  \foreach \i in \evenrow
  {\node[hd,text=ccol] at ({\vx+0.11},{\cy{\i}}) {\(\symd\)};}
  \foreach \i in \greyrows
  {\draw[black!45,line width=0.6pt] ({\vx},{\cy{\i}}) --
    ({\vx+0.22},{\cy{\i}});}

  % --------------------------------------------------- the columnwise reading
  \foreach \j in \qsg
  {\node[hd,text=qcol,anchor=north] at ({\cx{\j}},{-\gh-0.20})
    {\(\ilist[\j]\)};}

  \node[note,anchor=west] at ({\lx},{\yla})
  {\(M_\gamma = Z_2X_4\) occupies \(\gamma = \{4,7\}\),
    but \(X_i\) and \(Z_i\) anticommute,};
  \node[note,anchor=west,text=qcol] at ({\lx},{\ylb})
  {so it reads \(\qs{\gamma} = \gamma^{\times} = \{3,8\}\)};
  \node[note,anchor=west] at ({\lx},{\ylc})
  {\(\ilist[3]\symd\ilist[8]
    = \{\nu_2,\nu_5,{\color{ccol}\nu_7}\}
    \symd \{\nu_4,\nu_6,{\color{ccol}\nu_7}\}
    = {\color{acol}\{\nu_2,\nu_4,\nu_5,\nu_6\} = \anti(\gamma)}\)};

  % ------------------------------------------ the legend, right of the footer
  %% Its three lines share the footer baselines and end flush with the right
  %% edge of the row bands, so the band reads as one block and not as a column
  %% of text hung beside the grid.
  \fill[acol] ({\kx},{\yla-0.10}) -- ({\kx},{\yla+0.10})
  -- ({\kx+0.22},{\yla}) -- cycle;
  \node[hd,anchor=west,text=acol] at ({\kx+0.32},{\yla})
  {odd overlap: anticommutes};
  \node[hd,text=ccol] at ({\kx+0.11},{\ylb}) {\(\symd\)};
  \node[hd,anchor=west,text=ccol] at ({\kx+0.32},{\ylb})
  {even overlap: commutes, cancels};
  \draw[black!45,line width=0.6pt] ({\kx},{\ylc}) -- ({\kx+0.22},{\ylc});
  \node[hd,anchor=west,text=black!55] at ({\kx+0.32},{\ylc})
  {no overlap: commutes, never examined};

\end{tikzpicture}
  \caption{\label{fig:pauli-query} \textbf{Pauli anticommutation index.}
    The grid has \(N=4\) sites, cutoff \(d=2\), and \(K=7\) tracked monomials.
    The generator \(M_\gamma=Z_2X_4\) occupies \(\gamma=\{4,7\}\) but reads \(\qs{\gamma}=\gamma^{\times}=\{3,8\}\), the lists of \(X_2\) and \(Z_4\).
    Black cells mark occupied positions.
    The exchange finds \(\nu_4=Z_4\) and \(\nu_6=X_3Z_4\), which share no position with \(\gamma\) but anticommute with it, while \(\nu_7=X_2Z_4\) appears in both lists and cancels.}
\end{figure}

The index is the same in both bases. Only the positions a generator reads differ:
\begin{equation}
  \qs{\gamma}=
  \begin{cases}
    \gamma              & \text{Majorana, } |\gamma| \text{ even}, \\
    \gamma^{\times}     & \text{Pauli},                            \\
    [2N]\setminus\gamma & \text{Majorana, } |\gamma| \text{ odd}.
  \end{cases}
  \label{eq:query-support-cases}
\end{equation}
An even-length Majorana generator reads its own occupied positions because Eq.~\eqref{eq:majorana-commutation} makes anticommutation the parity of \(|\nu\cap\gamma|\).
In the Pauli basis, the exchange \(\gamma\mapsto\gamma^{\times}\) swaps the \(X_i\) and \(Z_i\) positions of each site, so an \(X_i\) in the generator reads the \(Z_i\) list, a \(Z_i\) reads the \(X_i\) list, and a \(Y_i\) reads both.
A site then contributes an odd overlap exactly when the two local factors anticommute.
For example, \(Y_i\) against \(Y_i\) gives two hits, which cancel.
Hence \(|\nu\cap\gamma^{\times}|\) has the parity of the number of sites with anticommuting factors, as in Eq.~\eqref{eq:pauli-commutation}.
For a bounded-weight Pauli generator, the lookup reads \(O(1)\) lists and costs \(O(\imax)\).
For an odd-length Majorana generator,
\[
  \{M_\nu,M_\gamma\}=0
  \iff |\nu|+|\nu\cap\gamma|\text{ is odd}
  \iff |\nu\setminus\gamma|\text{ is odd}.
\]
The last condition reads the complement \([2N]\setminus\gamma\), so the exact lookup can traverse the whole index and costs \(\Theta(K)\) in the worst case.

\section{Propagation algorithm}\label{app:propagation}

We formalise the propagation step of Section~\ref{sec:background}.
Given a Hermitian operator \(H=\sum_\nu c_\nu M_\nu\), a state \(\varrho\), and a circuit
\begin{equation}
  \mathcal C_L(\theta)=U_L(\theta_L)\cdots U_1(\theta_1),
  \qquad
  U_\ell(\theta_\ell)=e^{-\iu\theta_\ell M_{\gamma_\ell}/2},
\end{equation}
the target expectation value is
\begin{equation}
  f_L(\theta)=\Tr[\varrho\,  \mathcal C_L(\theta)^\dagger H \mathcal C_L(\theta)].
\end{equation}
We back-propagate \(H\) in the Heisenberg picture, which suits observables with a low-degree expansion.
Forward propagation of \(\varrho\) is equivalent up to truncation.

\code{monoprop} limits branching with a degree cutoff \(d\) (\(K=O(N^d)\)) and a coefficient threshold \(\epsilon\).
A branch onto \(\nu'\) from \(\nu\) at angle \(\theta\) is \emph{truncated} unless \(|c_\nu\sin\theta|>\epsilon\) and either \(\dg{\nu'}\le d\) or \(\nu'\) is diagonal (Appendix~\ref{app:reference-states}).
Diagonal partners are exempt from the degree cutoff because they are exactly the monomials that contribute to expectation values in a product reference state.
They are usually few in the simulated regimes, but including them improves accuracy.
Because exempt diagonal monomials can exceed degree \(d\), the occupancy bounds of Appendix~\ref{app:occupancy} hold strictly only with \(d\) replaced by the largest tracked degree.

A gate with generator \(M_\gamma\) and angle \(\theta\) leaves a commuting monomial unchanged and branches an anticommuting one:
\begin{equation}
  c_\nu M_\nu \xrightarrow{U}
  \begin{cases}
    c_\nu M_\nu                                                                     & \text{if } [M_\nu, M_{\gamma}] = 0,                             \\
    c_\nu\cos\theta\, M_\nu + \phi_\gamma(\nu)\,c_\nu\sin\theta\,M_{\gamma\symd\nu} & \text{if } \{M_\nu, M_{\gamma}\} = 0 \text{ and not truncated}, \\
    c_\nu\cos\theta\, M_\nu                                                         & \text{if } \{M_\nu, M_{\gamma}\} = 0 \text{ and truncated},
  \end{cases}
  \label{eq:mp-dispatch}
\end{equation}
where \(M_\gamma M_\nu=\chi(\gamma,\nu)M_{\gamma\symd\nu}\) by Eq.~\eqref{eq:monomial-product}, and the \emph{rotation sign} of an anticommuting \(\nu\) is
\begin{equation}
  \phi_\gamma(\nu)\coloneq\iu\,\chi(\gamma,\nu)\in\{\pm1\},
  \qquad
  \phi_\gamma(\nu\symd\gamma)=-\phi_\gamma(\nu).
  \label{eq:rotation-sign}
\end{equation}
The branch is added to the partner's coefficient, \(c_{\gamma\symd\nu} \rightarrow c_{\gamma\symd\nu} + \phi_\gamma(\nu)\,c_\nu\sin\theta\), with \(c_{\gamma\symd\nu}=0\) if the partner is not yet tracked.
A tracked partner is itself anticommuting and branches back onto \(\nu\) in the same way, so the two contributions add.
Starting from \(H\), backward propagation applies the gates in reverse circuit order, \(U_L\) first and \(U_1\) last, combining coefficients of the same monomial after each gate, and ends with \(H_L\coloneq\mathcal C_L(\theta)^\dagger H\,\mathcal C_L(\theta)\) up to truncation.

Algorithm~\ref{alg:expectation-value} computes the expectation value by back-propagating \(H\) and contracting the result with the product reference state \(\varrho\).
The operator is held in the data structures of Section~\ref{sec:gate-update}, including the anticommutation index \(\ilist\).
The cutoff \(d\) and threshold \(\epsilon\) are fixed as input.
The pseudocode is for a single process holding the operator.
For multiple processes, each process stores its owned monomials and data structures, performs the lookup locally, and batches partner requests to the owning rank (Appendix~\ref{app:rank-routing}).
All ranks follow the same gate order, and their local expectation-value contributions are combined.

Algorithm~\ref{alg:apply-gate} updates only the anticommuting terms returned by the index lookup.

\begin{algorithm}[H]
  \caption{\textsc{ExpectationValue}\((H,\mathcal C_L(\theta),\varrho;d,\epsilon)\)}
  \label{alg:expectation-value}
  \begin{algorithmic}[1]
    \State initialise the monomial vector, hash table, and \(\ilist\) from \(H\)
    \For{\(\ell=L,\ldots,1\)}
    \State \(\Call{ApplyGate}{\gamma_\ell,\theta_\ell}\)
    \EndFor
    \State \Return \(\sum_{\nu\ \mathrm{diagonal}} c_\nu \prod_{i\in\sites(\nu)}\zeta_i\) \Comment{Appendix~\ref{app:reference-states}}
  \end{algorithmic}
\end{algorithm}

\subsection{Leader--follower split}\label{app:leader-follower}
For a gate generator \(\gamma\), the partner of a tracked monomial \(\nu\) is \(\nu'=\nu\symd\gamma\), and the relation is symmetric, \(\nu = \nu' \symd \gamma\).
The symmetric difference flips membership of the positions in \(\gamma\) and leaves all other positions unchanged.
Choose any element of \(\gamma\) as the pivot, for example \(j^\star=\min\gamma\).
Because \(j^\star\in\gamma\), the gate flips its membership:
\[
  j^\star\in\nu' \quad\Longleftrightarrow\quad j^\star\notin\nu.
\]
Thus two partners cannot both contain the pivot or both omit it.
Call the endpoint without the pivot the \emph{leader} and the endpoint with it the \emph{follower}.

For example, take Majorana monomials \(\gamma=\{1,2\}\) and \(\nu=\{2,4\}\), with \(j^\star=1\), which gives
\[
  \nu'=\underbrace{\{2,4\}}_{\nu\text{ (leader)}}\symd
  \underbrace{\{1,2\}}_{\gamma}
  =\underbrace{\{1,4\}}_{\text{follower}}.
\]
Applying the same symmetric difference again returns the leader: \(\{1,4\}\symd\{1,2\}=\{2,4\}\).

The index gives this split directly.
Since \(M_{\nu'}\propto M_\gamma M_\nu\) anticommutes with \(M_\gamma\) exactly when \(M_\nu\) does, both endpoints of an anticommuting pair lie in \(\anti(\gamma)\), which splits into the leaders \(\anti(\gamma)\setminus\ilist[j^\star]\) and the followers \(\anti(\gamma)\cap\ilist[j^\star]\).
We process the leaders first.
If a leader's partner is tracked, it is a follower, and we rotate the pair and mark the follower.
Otherwise, we add the partner unless it is truncated.
We then treat the unmarked followers, whose leaders are absent, in the same way.
Each anticommuting monomial is therefore probed at most once, and each tracked pair is rotated once.

\begin{algorithm}[H]
  \caption{\textsc{ApplyGate}\((\gamma,\theta)\)}
  \label{alg:apply-gate}
  \begin{algorithmic}[1]
    \State \(\anti(\gamma)\gets\bigsymd_{j\in\qs{\gamma}}\ilist[j]\), \quad \(j^\star\gets\min\gamma\) \Comment{lookup (Theorem~\ref{thm:xor})}
    \For{\(\nu\in\anti(\gamma)\setminus\ilist[j^\star]\), then unmarked \(\nu\in\anti(\gamma)\cap\ilist[j^\star]\)} \Comment{leaders, then followers}
    \State \(\nu'\gets\nu\symd\gamma\), \quad \(\phi\gets\phi_\gamma(\nu)\) \Comment{Eq.~\eqref{eq:rotation-sign}}
    \If{\(\nu'\) is tracked} \Comment{one probe, one rotation per pair}
    \State mark \(\nu'\); \quad \((c_\nu,c_{\nu'})\gets(c_\nu\cos\theta-\phi\,c_{\nu'}\sin\theta,\ c_{\nu'}\cos\theta+\phi\,c_\nu\sin\theta)\)
    \ElsIf{\(|c_\nu\sin\theta|>\epsilon\) and (\(\dg{\nu'}\le d\) or \(\nu'\) diagonal)} \Comment{branch}
    \State append \(\nu'\) to the monomial vector, hash table and \(\ilist\) with \(c_{\nu'}\gets\phi\,c_\nu\sin\theta\); \quad \(c_\nu\gets c_\nu\cos\theta\)
    \Else \Comment{truncated}
    \State \(c_\nu\gets c_\nu\cos\theta\)
    \EndIf
    \EndFor
  \end{algorithmic}
\end{algorithm}

\section{Partitioning and Linear Routing}\label{app:rank-routing}
In both bases, monomials are assigned to \(R\) MPI ranks and, within each rank, to \(P\) shared-memory partitions.
Each partition holds the data structures of Section~\ref{sec:gate-update} for the monomials assigned to it.
We usually assign one partition per core, so each rank owns \(P\) partitions.
For a gate generator \(\gamma\), \(\nu \in \anti(\gamma)\) on rank \(r\) and partition \(p\) must find its partner \(\nu\symd\gamma\) on rank \(r'\) and partition \(p'\).
An ordinary hash scatters the partners of a rank's monomials across many ranks, whereas the map below sends every partner request from rank \(r\) to a single destination determined by \(r\) and \(\gamma\).

\begin{definition}[Linear rank map]\label{def:linear-rank-hash}
  Let \(R=2^q\) and represent rank numbers by \(q\) bits.
  Give each position \(j\in[2N]\) a fixed \(q\)-bit label \(\xi_j\), identical across ranks but not necessarily distinct across positions.
  The owner rank of monomial \(\nu\) is the bitwise XOR of its positions' labels,
  \begin{equation}
    \rank(\nu)=\bigoplus_{j\in\nu}\xi_j.
    \label{eq:linear-rank-hash}
  \end{equation}
\end{definition}
The map is linear with respect to symmetric difference:
\begin{equation}
  \rank(\nu\symd\gamma)=\rank(\nu)\oplus\rank(\gamma).
  \label{eq:linear-rank-shift}
\end{equation}
Every partner request from rank \(r\) goes to \(r'=r\oplus\rank(\gamma)\).
When \(\rank(\gamma)=0\), no inter-rank communication is required.
Any labels satisfy this identity, but their choice sets the load balance across ranks.
We draw them from a fixed seed.
Within each rank, \(\splitmix\) assigns local partitions independently, giving the global partition number
\[
  \rank(\nu)P+\bigl(\splitmix(\nu)\bmod P\bigr).
\]
This spreads local work without changing the single destination rank.

The bitwise partitioning of parallel iterative qubit coupled cluster (iQCC) can be viewed as a linear hash, which assigns each of \(q\) selected positions a different unit-bit label and all others zero~\cite{hosseini_parallel_iqcc_2026}.
This fixes the partner partition, but Hosseini et al.\ assign partitions to processors dynamically for load balance, so a processor may still communicate with several others per gate.
Broers, Sun, and Yunoki instead sum blocks of the Pauli bit string modulo the process count~\cite{broers_orqa_2025}.
An XOR flip can increase or decrease that sum, giving up to \(2^{2\dg{\gamma}}+1\) destinations.

\section{Surrogate graphs and gradients}\label{app:surrogate}

The surrogate graph records the partner pairs and signs met during propagation (Appendix~\ref{app:propagation}) and replays them to re-evaluate the expectation value and gradient at new angles.
It is a sequence \(\mathcal G=(\mathcal L_\ell)_{\ell=1}^{L}\) of layers, one per gate (Appendix~\ref{app:layer-contents}).
Layers are numbered in propagation order, so layer \(\ell\) belongs to \(U_{L+1-\ell}\), the \(\ell\)-th gate propagated.
In this appendix only, \(\gamma_\ell\) denotes that gate's generator.
Its angle is \(\theta_{p_\ell}\), and \(p_\ell\) lets several gates share one parameter.

\subsection{Replay as a linear map}\label{app:surrogate-map}

Layer \(\ell\) stores a map \(T_\ell(\theta_{p_\ell})\) that updates the coefficient vector \(c\) by a gate.
The map multiplies anticommuting coefficients by \(\cos\theta_{p_\ell}\) and adds \(\sin\theta_{p_\ell}\) corrections to untruncated pairs.
The surrogate expectation is therefore
\begin{equation}
  f_L(\theta) \approx \widetilde f(\theta)
  = w^{\mathsf T}T_L(\theta_{p_L})\cdots T_1(\theta_{p_1})c_0,
  \label{eq:surrogate-linear-map}
\end{equation}
where \(c_0\) is the initial coefficient vector and \(w\) holds the rescaled reference-state coefficients \(2^Nw_\nu\) of Appendix~\ref{app:reference-states}, which are \(\prod_{i\in\sites(\nu)}\zeta_i\) on tracked diagonal monomials and zero elsewhere.
The map also gives the analytical gradient.
Let \(c_{\ell-1}\) be the coefficient vector propagated through preceding layers, and let \(w_\ell\) be the contributing state coefficients pulled back through the following ones, so that \(w_L=w\) and \(w_{\ell-1}=T_\ell^{\mathsf T}w_\ell\).
For a parameter used by a single layer \(\ell\), the derivative is
\begin{equation}
  \partial_{\theta_{p_\ell}}\widetilde f
  =
  w_\ell^{\mathsf T}\bigl(\partial_{\theta_{p_\ell}}T_\ell\bigr)c_{\ell-1}.
  \label{eq:surrogate-gradient}
\end{equation}
If multiple layers share the same gate parameter, their contributions are added.
Appendix~\ref{app:analytic-gradient} evaluates Eq.~\eqref{eq:surrogate-gradient} for every parameter in one forward and one backward sweep.

\subsection{Layer structure}\label{app:layer-contents}
Layer \(\ell\) stores
\begin{equation}
  \mathcal L_\ell
  = \bigl(p_\ell,\;\mathcal P_\ell,\;\Phi_\ell\bigr),
  \label{eq:layer}
\end{equation}
where \(\mathcal P_\ell=\{(s_{\ell b},t_{\ell b})\}_b\) lists the coefficient-vector positions of each untruncated partner pair, and \(\Phi_\ell=\{\phi_{\ell b}\}_b\) their rotation signs (Eq.~\eqref{eq:rotation-sign}).
For one stored pair \((s,t,\phi)\), the coefficients rotate as
\begin{equation}
  \begin{pmatrix}c_s'\\c_t'\end{pmatrix}
  =
  \begin{pmatrix}
    \cos\theta     & -\phi\sin\theta \\
    \phi\sin\theta & \cos\theta
  \end{pmatrix}
  \begin{pmatrix}c_s\\c_t\end{pmatrix},
  \qquad \theta=\theta_{p_\ell}.
  \label{eq:surrogate-pair-rotation}
\end{equation}
An anticommuting monomial whose partner was truncated is only multiplied by \(\cos\theta\).
By Theorem~\ref{thm:xor}, a layer stores at most \(|\anti(\gamma_\ell)|=O(\imax)\) pairs, which is \(O(K/N)\) for an evenly supported operator.
The layer does not store \(\anti(\gamma_\ell)\), which receives the cosine factor, but looks it up from the index (Eq.~\eqref{eq:index-query}) at each replay.
Algorithm~\ref{alg:apply-surrogate-layer} records the pair endpoints before applying the cosine factor to \(\anti(\gamma_\ell)\) and the sine corrections to the stored pairs.

\begin{algorithm}[H]
  \caption{\textsc{ApplyLayer}\((c,\mathcal L_\ell,\theta,\anti(\gamma_\ell))\)}
  \label{alg:apply-surrogate-layer}
  \begin{algorithmic}[1]
    \For{\(b=1,\ldots,|\mathcal P_\ell|\)}
    \State \((s_b,t_b)\gets\mathcal P_\ell[b]\), \quad \(\phi_b\gets\Phi_\ell[b]\)
    \State \((u_b,v_b)\gets(c_{s_b},c_{t_b})\) \Comment{copy the endpoints before they are modified}
    \EndFor
    \For{\(\nu\in\anti(\gamma_\ell)\)}
    \State \(c_\nu\gets c_\nu \cos\theta \)
    \EndFor
    \For{\(b=1,\ldots,|\mathcal P_\ell|\)}
    \State \(c_{s_b}\gets c_{s_b}-\phi_b v_b\sin\theta\)
    \State \(c_{t_b}\gets c_{t_b}+\phi_b u_b\sin\theta\)
    \EndFor
    \State \Return \(c\)
  \end{algorithmic}
\end{algorithm}

\subsection{Analytical gradient}\label{app:analytic-gradient}

In Eq.~\eqref{eq:surrogate-gradient}, the differentiated pair rotation is
\begin{equation}
  \partial_{\theta}
  \begin{pmatrix}
    \cos\theta     & -\phi\sin\theta \\
    \phi\sin\theta & \cos\theta
  \end{pmatrix}
  =
  \begin{pmatrix}
    -\sin\theta    & -\phi\cos\theta \\
    \phi\cos\theta & -\sin\theta
  \end{pmatrix}.
  \label{eq:surrogate-pair-derivative}
\end{equation}
Algorithm~\ref{alg:surrogate-gradient} evaluates the expectation value and all gradients with one forward and one reverse sweep.
Here \(\textsc{DifferentiateLayer}\) applies Eq.~\eqref{eq:surrogate-pair-derivative} to untruncated pairs, multiplies cosine-only entries by \(-\sin\theta\), sets all entries outside \(\anti(\gamma_\ell)\) to zero, and returns the resulting vector.
The operation \(\textsc{UndoLayer}\) rotates each stored pair by \(-\theta\) and divides cosine-only entries by \(\cos\theta\), recovering the coefficient vector before the layer.
This in-place recovery assumes \(\cos\theta_{p_\ell}\neq0\).
Otherwise, \code{monoprop} stores the coefficients before the layer.
Each pair block of \(T_\ell\) (Eq.~\eqref{eq:surrogate-pair-rotation}) is a rotation, whose transpose is the rotation by \(-\theta\).
The remaining entries are diagonal and even in \(\theta\), so \(T_\ell(\theta)^{\mathsf T}=T_\ell(-\theta)\).
The pull-back \(w_{\ell-1}=T_\ell^{\mathsf T}w_\ell\) is therefore a call to \textsc{ApplyLayer} at angle \(-\theta\).
This is the transpose, not the inverse, because the cosine-only entries make \(T_\ell\) non-orthogonal.
In Algorithm~\ref{alg:surrogate-gradient}, \(\bar w\) holds the pulled-back \(w_\ell\), starting from \(w\).

\begin{algorithm}[H]
  \caption{\textsc{EvaluateAndGradient}\((c_0,w,\mathcal G,\theta,\{\anti(\gamma_\ell)\})\)}
  \label{alg:surrogate-gradient}
  \begin{algorithmic}[1]
    \State \(c\gets c_0\)
    \For{\(\ell=1,\ldots,L\)}
    \State \(c\gets\Call{ApplyLayer}{c,\mathcal L_\ell,\theta_{p_\ell},\anti(\gamma_\ell)}\)
    \EndFor
    \State \(\widetilde f\gets w^{\mathsf T}c\), \quad \(\bar w\gets w\), \quad \(\nabla\widetilde f\gets0\)
    \For{\(\ell=L,\ldots,1\)}
    \State \(c\gets\Call{UndoLayer}{c,\mathcal L_\ell,\theta_{p_\ell},\anti(\gamma_\ell)}\) \Comment{recover \(c_{\ell-1}\)}
    \State \(\dot c\gets\Call{DifferentiateLayer}{c,\mathcal L_\ell,\theta_{p_\ell},\anti(\gamma_\ell)}\)
    \State \(\partial_{\theta_{p_\ell}}\widetilde f\gets\partial_{\theta_{p_\ell}}\widetilde f+\bar w^{\mathsf T}\dot c\)
    \State \(\bar w\gets\Call{ApplyLayer}{\bar w,\mathcal L_\ell,-\theta_{p_\ell},\anti(\gamma_\ell)}\) \Comment{apply \(T_\ell^{\mathsf T}\)}
    \EndFor
    \State \Return \((\widetilde f,\nabla\widetilde f)\)
  \end{algorithmic}
\end{algorithm}

\section{Generating optimal gates and angles}\label{app:gate-selection}

This appendix proves Theorem~\ref{thm:energy-selection} in the Majorana basis with a Fock reference.
Appendix~\ref{app:one-gate-energy} reduces scoring to two coefficients \(A_\gamma\) and \(B_\gamma\), Lemmas~\ref{lem:flip-table} and~\ref{lem:gradients} compute them, and Appendix~\ref{app:generate-proof} combines them.

An electronic-structure Hamiltonian with one- and two-body terms has even Majorana length of at most four~\cite{mcardle_chemistry_2020}.
Propagating it through gates with even-length generators preserves its parity but can increase the length of its monomials, up to the cutoff \(d\).
Let \(H_{L-1}=\sum_{\nu\in\mathcal K_d}c_\nu M_\nu\), of even parity and degree at most \(d\), be the observable back-propagated through the \(L-1\) gates already in the circuit, so that the candidate \(L\)-th gate acts directly on the reference.
Let \(\ket n = \ket{n_1, n_2, \dots, n_N}\), with \(n_i\in\{0,1\}\), be the Fock reference.
A Majorana monomial \(\nu\) in \(H_{L-1}\) has a non-zero Fock expectation when it is \emph{diagonal}, that is, when \(M_\nu=\prod_{i\in S} \iu m_{2i-1} m_{2i}\) for a mode set \(S\) (Appendix~\ref{app:reference-states}).

The candidate gates are \(U_\gamma(\theta)=e^{-\iu\theta M_\gamma/2}\) with \(|\gamma|\) even.
If \(\gamma\) contains both Majoranas of a mode, that pair acts on \(\ket n\) as a sign, which reversing the angle absorbs, so on the reference the gate equals one whose generator has degree two lower.
Thus, for each even \(\gdeg\), it suffices to consider generators with one Majorana on each of \(\gdeg\) distinct modes, so that \(|\gamma|=|\sites(\gamma)|=\gdeg\).
At degree \(\gdeg\) the pool has \(2^\gdeg\binom N\gdeg=O(N^\gdeg)\) candidates.
Scoring this pool directly costs \(O(N^{d+\gdeg})\) (Section~\ref{sec:selection}).

\subsection{Energy of one gate}\label{app:one-gate-energy}

Write \(E(n')=\langle n'|H_{L-1}|n'\rangle\) for the energy of an occupation string \(n'\), \(\flip{n}{\Lambda}\) for \(n\) with the occupations of the modes in \(\Lambda\) flipped, \(f_\gamma(\theta)=\langle n|U_\gamma^\dagger H_{L-1}U_\gamma|n\rangle\), and \(\langle\cdot\rangle_n\coloneq\langle n|\cdot|n\rangle\).
Since \(M_\gamma^2=I\), the gate is \(\cos(\theta/2)I-\iu\sin(\theta/2)M_\gamma\).
Split the observable as \(H_{L-1}=H_{\anti(\gamma)}+H_{\com(\gamma)}\), the sums of the tracked monomials in \(\anti(\gamma)\) and in its complement \(\com(\gamma)=\mathcal K_d\setminus\anti(\gamma)\), which anticommute and commute with \(M_\gamma\), respectively.
Conjugation by \(M_\gamma\) negates the first and fixes the second, so
\begin{equation}
  U_\gamma^\dagger H_{L-1}U_\gamma
  =H_{\com(\gamma)}+\cos\theta\,H_{\anti(\gamma)}
  +\iu\sin\theta\,M_\gamma H_{\anti(\gamma)}.
  \label{eq:conjugated-split}
\end{equation}
Its reference expectation is the single-angle sinusoid~\cite{nakanishi_sequential_2020,ostaszewski_structure_2021}
\begin{equation}
  f_\gamma(\theta)=A_\gamma\cos\theta+B_\gamma\sin\theta+C_\gamma,
  \qquad
  A_\gamma=\langle H_{\anti(\gamma)}\rangle_n,\quad
  B_\gamma=\iu\langle M_\gamma H_{\anti(\gamma)}\rangle_n,\quad
  C_\gamma=\langle H_{\com(\gamma)}\rangle_n.
  \label{eq:one-angle}
\end{equation}
The operator \(M_\gamma H_{\anti(\gamma)}\) is anti-Hermitian, so \(B_\gamma\) is real, and it is the gradient \(f_\gamma'(0)\).
Because \(M_\gamma\) flips the occupations of \(\sites(\gamma)\), \(\langle M_\gamma H_{L-1}M_\gamma\rangle_n=E(\flip{n}{\sites(\gamma)})=C_\gamma-A_\gamma\).
Since \(E(n)=A_\gamma+C_\gamma\),
\begin{equation}
  A_\gamma=\tfrac12\bigl[E(n)-E(\flip{n}{\sites(\gamma)})\bigr].
  \label{eq:A-as-flip}
\end{equation}
The oscillating part has amplitude \(\sqrt{A_\gamma^2+B_\gamma^2}\), so the optimal angle and energy decrease are
\begin{equation}
  \theta_\gamma^*=\operatorname{atan2}(-B_\gamma,-A_\gamma),
  \qquad
  \Delta E_\gamma=A_\gamma+\sqrt{A_\gamma^2+B_\gamma^2},
  \label{eq:optimal-angle}
\end{equation}
with \(\theta_\gamma^*=0\) when \(A_\gamma=B_\gamma=0\).
A candidate with \(B_\gamma=0\) and \(A_\gamma>0\) still lowers the energy, at \(\theta=\pi\), which gradient-only selection misses.
Scoring the pool thus reduces to computing \(A_\gamma\) and \(B_\gamma\) for every candidate.

\subsection{Flip energies for all candidates}\label{app:flip-table}

\begin{lemma}[Flip table]\label{lem:flip-table}
  After expected \(O(N^{d/2})\) preprocessing, \(E(\flip{n}{\Lambda})\), and hence \(A_\gamma\) by Eq.~\eqref{eq:A-as-flip}, is available for every mode set \(\Lambda\) with \(|\Lambda|\le\gdeg\) from \(2^{|\Lambda|}\) table entries.
\end{lemma}
\begin{proof}
  For each tracked diagonal monomial \(\delta\), write \(a_\delta=c_\delta\langle M_\delta\rangle_n\), and define
  \[
    W(\Lambda)=\sum_{\substack{\delta\text{ diagonal, tracked}\\\Lambda\subseteq\sites(\delta)}}a_\delta,
    \qquad |\Lambda|\le\gdeg,
  \]
  so that \(W(\varnothing)=E(n)\).
  A diagonal monomial within the cutoff occupies at most \(d/2\) modes, so enumerating the \(O(N^{d/2})\) such monomials, probing each, and adding each hit to its at most \(2^{d/2}=O(1)\) mode subsets takes expected \(O(N^{d/2})\).
  Flipping the occupations of \(\Lambda\) multiplies \(a_\delta\) by \((-1)^{|\sites(\delta)\cap\Lambda|}=\sum_{\Lambda'\subseteq\sites(\delta)\cap\Lambda}(-2)^{|\Lambda'|}\), and exchanging the sums gives
  \begin{equation}
    E(\flip{n}{\Lambda})=\sum_{\Lambda'\subseteq\Lambda}(-2)^{|\Lambda'|}W(\Lambda').
    \label{eq:flip-table}
  \end{equation}
  For two modes, for example, \(E(\flip{n}{\{i,i'\}})=E(n)-2W(\{i\})-2W(\{i'\})+4W(\{i,i'\})\).
\end{proof}

\subsection{Gradients from the state side}\label{app:gradient-enumeration}

\begin{figure}[t]
  \centering
  \begin{tikzpicture}[x=0.62cm, y=0.62cm,
      cell/.style={draw=black!35, line width=0.3pt},
      on/.style={draw=black!35, line width=0.3pt, fill=black!72},
      lab/.style={font=\small, anchor=east},
      note/.style={font=\footnotesize, anchor=west},
      hd/.style={font=\footnotesize}]
    \foreach \x in {0,1,2.4,3.4,4.8,5.8}
    \foreach \y in {0,-1,-2}
    \draw[cell] (\x,\y) rectangle ++(0.9,0.8);
    \foreach \x in {1,2.4} \draw[on] (\x,0) rectangle ++(0.9,0.8);
    \foreach \x in {0,2.4,4.8,5.8} \draw[on] (\x,-1) rectangle ++(0.9,0.8);
    \foreach \x in {0,1,4.8,5.8} \draw[on] (\x,-2) rectangle ++(0.9,0.8);
    \node[lab] at (-0.2,0.4) {\(\gamma\)};
    \node[lab] at (-0.2,-0.6) {\(\nu\)};
    \node[lab] at (-0.2,-1.6) {\(\delta=\gamma\symd\nu\)};
    \foreach \x/\j in {0/1,1/2,2.4/3,3.4/4,4.8/5,5.8/6}
    \node[hd] at (\x+0.45,1.15) {\(m_{\j}\)};
    \foreach \x/\t in {0/{\(Q\)},2.4/{\(V\)},4.8/{\(S\setminus Q\)}}
    \node[hd] at (\x+0.95,-2.6) {\t};
    \node[note] at (6.95,0.4) {one Majorana on each of \(\gdeg=2\) modes};
    \node[note] at (6.95,-0.6) {\(|\gamma\cap\nu|=\gdeg-\omega=1\), odd: anticommutes};
    \node[note] at (6.95,-1.6) {diagonal on \(2\) modes; union spans \(3=(d+\gdeg)/2\)};
  \end{tikzpicture}
  \caption{\label{fig:state-side} \textbf{A compatible pair, generated from the state side}, for \(d=4\) and \(\gdeg=2\).
    Rows are the monomials \(\gamma=\{2,3\}\), \(\nu=\{1,3,5,6\}\), and \(\delta=\{1,2,5,6\}\) on three modes, with black cells marking occupied positions.
    The diagonal \(\delta\) occupies the modes \(S\) and shares \(\omega=1\) of them, \(Q\), with the candidate, where \(\nu\) holds the other Majorana; the candidate's remaining mode \(V\) lies outside \(S\), where \(\nu\) holds the same Majorana, and \(\nu\) holds both Majoranas on \(S\setminus Q\).
    Choosing \(\delta\), \(Q\), \(V\), and one Majorana per candidate mode fixes the pair.}
\end{figure}
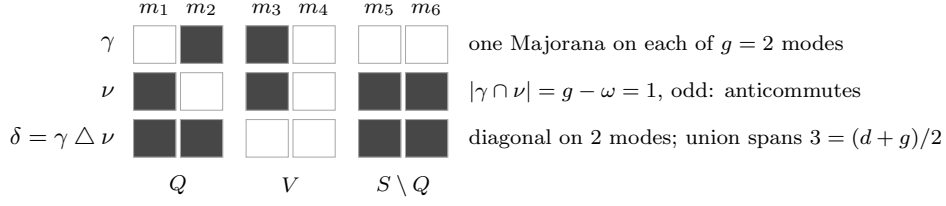

\begin{lemma}[State-side enumeration]\label{lem:gradients}
  The gradients \(B_\gamma\) of all candidates of degree \(\gdeg\) can be computed in expected \(O(N^{(d+\gdeg)/2})\) time, without scanning the \(K\) tracked monomials.
\end{lemma}
\begin{proof}
  Differentiating the gate gives
  \begin{equation}
    B_\gamma=\frac{\iu}{2}\langle n|[M_\gamma,H_{L-1}]|n\rangle
    =\sum_{\nu\in\anti(\gamma)}
    \phi_\gamma(\nu)\,c_\nu\langle M_{\gamma\symd\nu}\rangle_n,
    \label{eq:gradient-sum}
  \end{equation}
  with the rotation sign of Eq.~\eqref{eq:rotation-sign}.
  A monomial contributes only when \(\delta=\gamma\symd\nu\) is diagonal and \(|\gamma\cap\nu|\) is odd, and we call such a pair \((\gamma,\nu)\) \emph{compatible}.
  We generate compatible pairs from \(\delta\) (Fig.~\ref{fig:state-side}).
  Let \(S=\sites(\delta)\), let \(Q=\sites(\gamma)\cap S\) with \(\omega=|Q|\), and let \(V=\sites(\gamma)\setminus S\).
  Then \(\nu=\gamma\symd\delta\) holds the other Majorana on each mode of \(Q\), the same Majorana as \(\gamma\) on each mode of \(V\), and both Majoranas on each mode of \(S\setminus Q\).
  Hence \(|\gamma\cap\nu|=\gdeg-\omega\), which, as \(\gdeg\) is even, is odd exactly when \(\omega\) is odd, and
  \begin{equation}
    |\nu|=\gdeg+2(|S|-\omega),
    \qquad
    |S\cup V|=|S|+\gdeg-\omega=\frac{\gdeg+|\nu|}{2}\le\frac{d+\gdeg}{2}.
    \label{eq:selection-mode-count}
  \end{equation}
  So each compatible pair is fixed by a mode set \(S\), an odd \(\omega\), the sets \(Q\subseteq S\) and \(V\subseteq[N]\setminus S\) of sizes \(\omega\) and \(\gdeg-\omega\), and one Majorana on each mode of \(Q\cup V\), subject to the cutoff \(\gdeg+2|S|-2\omega\le d\).
  Conversely, a pair determines \(\delta\) and hence this data, so every compatible pair is generated exactly once.
  The modes involved number at most \((d+\gdeg)/2\), giving \(O(N^{(d+\gdeg)/2})\) choices with constantly many configurations each, and one expected constant-time hash probe for \(c_\nu\) per choice.
\end{proof}

\begin{algorithm}[H]
  \caption{\textsc{GenerateGates}\((c,n;\gdeg)\)}
  \label{alg:generate-gates}
  \begin{algorithmic}[1]
    \Require coefficients \(c\), Fock state \(n\), \(N\) modes, even \(2\le\gdeg\le d\)
    \State initialise empty tables \(B,W\) and output list
    \Statex \Comment{1. Flip table (Lemma~\ref{lem:flip-table}): expected \(O(N^{d/2})\)}
    \ForAll{mode sets \(S\) with \(|S|\le\lfloor d/2\rfloor\)}
    \State \(\delta\gets\Call{Pairs}{S}\), \quad \(a_\delta\gets c_\delta\langle M_\delta\rangle_n\)
    \ForAll{\(\Lambda\subseteq S\) with \(|\Lambda|\le\gdeg\)}
    \State \(W(\Lambda)\gets W(\Lambda)+a_\delta\)
    \EndFor
    \EndFor
    \Statex \Comment{2. Gradients (Lemma~\ref{lem:gradients}): expected \(O(N^{(d+\gdeg)/2})\)}
    \ForAll{mode sets \(S\) with \(|S|\le\lfloor(d+\gdeg)/2\rfloor-1\)}
    \State \(\delta\gets\Call{Pairs}{S}\)
    \ForAll{odd \(1\le\omega\le\min(\gdeg,|S|)\) with \(\gdeg+2|S|-2\omega\le d\)}
    \ForAll{\(Q\subseteq S\), \(|Q|=\omega\), and \(V\subseteq[N]\setminus S\), \(|V|=\gdeg-\omega\)}
    \ForAll{\(\gamma\) choosing one position per mode in \(Q\cup V\)}
    \State \(\nu\gets\gamma\symd\delta\) \Comment{compatible by construction}
    \State \(B_\gamma\gets B_\gamma+\phi_\gamma(\nu)\,c_\nu\langle M_\delta\rangle_n\)
    \EndFor
    \EndFor
    \EndFor
    \EndFor
    \Statex \Comment{3. Score every candidate, Eq.~\eqref{eq:optimal-angle}: \(O(N^\gdeg)\)}
    \ForAll{mode sets \(\Lambda\) with \(|\Lambda|=\gdeg\)}
    \State \(E_{\mathrm{flip}}\gets\sum_{\Lambda'\subseteq\Lambda}(-2)^{|\Lambda'|}W(\Lambda')\) \Comment{\(2^\gdeg=O(1)\) entries}
    \ForAll{\(\gamma\) choosing one position per mode in \(\Lambda\)}
    \State \(A_\gamma\gets[W(\varnothing)-E_{\mathrm{flip}}]/2\)
    \State \(\theta_\gamma^*\gets0\)
    \If{\(A_\gamma\ne0\) or \(B_\gamma\ne0\)}
    \State \(\theta_\gamma^*\gets\operatorname{atan2}(-B_\gamma,-A_\gamma)\)
    \EndIf
    \State \(\Delta E_\gamma\gets A_\gamma+\sqrt{A_\gamma^2+B_\gamma^2}\)
    \State append \((\gamma,\theta_\gamma^*,\Delta E_\gamma)\) to output
    \EndFor
    \EndFor
    \State \Return output
  \end{algorithmic}
\end{algorithm}

\subsection{Algorithm and proof of \texorpdfstring{Theorem~\ref{thm:energy-selection}}{the selection theorem}}\label{app:generate-proof}

Algorithm~\ref{alg:generate-gates} runs the two lemmas and then scores every candidate.
Its input \(c\) is the coefficient hash table maintained by propagation.
Absent entries of \(c\), \(B\), and \(W\) read as zero.
For a mode set \(S\), \(\textsc{Pairs}(S)\) is the diagonal monomial containing both Majoranas of every mode in \(S\).
Choosing one position per mode of a mode set enumerates the generators with one Majorana on each of those modes.

\begin{proof}[Proof of Theorem~\ref{thm:energy-selection}]
  Fix \(d,\gdeg=O(1)\).
  For each even degree \(\gdeg'\le\gdeg\), Algorithm~\ref{alg:generate-gates} returns \((\gamma,\theta_\gamma^*,\Delta E_\gamma)\) for every candidate of degree \(\gdeg'\).
  Stage~1 builds the flip table once for all degrees, in expected \(O(N^{d/2})\) time by Lemma~\ref{lem:flip-table}.
  Stage~2 computes every \(B_\gamma\) in expected \(O(N^{(d+\gdeg')/2})\) by Lemma~\ref{lem:gradients}.
  Stage~3 reads \(2^{\gdeg'}\) table entries and applies Eq.~\eqref{eq:optimal-angle} for each of the \(O(N^{\gdeg'})\) candidates, taking unwritten gradients as zero.
  Summed over the constantly many degrees, and using \(\gdeg\le d\), the time is expected
  \[
    O\!\left(N^{(d+\gdeg)/2}+N^{d/2}+N^\gdeg\right)=O\!\left(N^{(d+\gdeg)/2}\right),
  \]
  with \(O(N^\gdeg)\) memory for the tables and output.
  By the reduction of the candidate pool, these candidates cover every even monomial gate of degree at most \(\gdeg\) that acts non-trivially on the reference.
  The next gate is the candidate with the largest \(\Delta E_\gamma\), and its energy is \(W(\varnothing)-\Delta E_\gamma\).
\end{proof}

\section{Benchmark models}\label{app:models}
\subsection{One-dimensional Hubbard model}\label{app:hubbard}
The distributed benchmarks of Fig.~\ref{fig:mpi-scaling} use an open Hubbard chain with 60 lattice sites and two spin species, giving \(N=120\) fermionic modes and Hamiltonian
\begin{equation}
  H = -t \sum_{i=1}^{59} \sum_{\sigma \in \{\uparrow, \downarrow\}}
  \bigl( a_{i\sigma}^{\dagger} a_{i+1,\sigma} + a_{i+1,\sigma}^{\dagger} a_{i\sigma} \bigr)
  + \mathcal U \sum_{i=1}^{60} n_{i\uparrow} n_{i\downarrow},
  \label{eq:hubbard}
\end{equation}
with \(t=1\), \(\mathcal U=-2\), zero chemical potential, and \(n_{i\sigma}=a_{i\sigma}^{\dagger}a_{i\sigma}\).
We propagate \(n_{46\uparrow}\) through one first-order Trotter step of duration \(\Delta t=0.2\), which applies each Hamiltonian term once, with the half-filled N\'eel state as reference.
The Majorana-length cutoff is \(d=10\), and we vary the coefficient threshold to control the number of tracked monomials.

\subsection{Kicked-Ising model}\label{app:kicked-ising}
Figure~\ref{fig:inverse-scaling}(a) uses an open kicked-Ising chain with \(N=32,64,\ldots,1024\) qubits.
Each layer applies a \(Z_iZ_{i+1}\) rotation on every nearest-neighbour bond, followed by an \(X_i\) rotation on every site, all at angle \(\pi/4\).
The observable is the magnetisation \(\sum_{i=1}^{N}Z_i\), back-propagated with Pauli-weight cutoff \(d=6\) and no coefficient truncation through seven, ten, and thirteen layers, tracking \(6.1\times10^{5}\)--\(2.5\times10^{7}\), \(2.7\times10^{6}\)--\(1.2\times10^{8}\), and \(7.2\times10^{6}\)--\(3.7\times10^{8}\) monomials, respectively, from \(N=32\) to 1024.

\subsection{Two-dimensional Hubbard model}\label{app:hubbard-2d}
Figure~\ref{fig:inverse-scaling}(b) uses the Hamiltonian of Eq.~\eqref{eq:hubbard} on open \(s\times s\) square lattices with side \(s=4\), 6, 8, 11, 16, and 22, giving \(N=2s^2=32,\ldots,968\) fermionic modes, with nearest-neighbour hopping \(t=1\), \(\mathcal U=1.5\), and zero chemical potential.
Each layer is a first-order Trotter step of duration \(\Delta t=0.07\), which applies the spin-up hopping on every bond, then the spin-down hopping on every bond, then the interaction on every site.
The bonds are applied in four classes: horizontal bonds starting at even \(x\), then odd \(x\), then vertical bonds starting at even \(y\), then odd \(y\).
Gates within a class commute, so the operator spreads by a bounded distance per layer.
The observable is the staggered magnetisation \(\sum_{i}(-1)^{x_i+y_i}(n_{i\uparrow}-n_{i\downarrow})\), back-propagated through three, four, and five layers, and the reference state is the N\'eel state with spin up on the sites of even \(x+y\).
There is no coefficient truncation.
Instead of the degree cutoff of Appendix~\ref{app:notation}, a monomial is kept if it acts on at most three modes, whatever its length, or if it is diagonal.
\code{MajoranaPropagation.jl} has no built-in truncation with this rule, so we pass it as a custom truncation function and have it truncate after each of the two Majorana rotations of a hopping gate (option \code{truncate\_each\_mr}), as \code{monoprop} does.
By default it would truncate once per hopping gate and keep more monomials.
With these settings, the two codes track identical monomials after every hopping gate.
Three layers track \(1.7\times10^{4}\) to \(3.0\times10^{7}\) monomials from \(N=32\) to 968, four layers \(1.7\times10^{4}\) to \(9.4\times10^{7}\), and five layers \(1.7\times10^{4}\) to \(2.0\times10^{8}\).
The cutoff saturates the small lattices, from three layers on the \(4\times4\) lattice and from four on the \(6\times6\) and \(8\times8\) lattices, so deeper layers there track the same monomials.

\subsection{Heavy-hex kicked-Ising circuit}\label{app:heavy-hex}
Figure~\ref{fig:graph-gradient} uses the 127-qubit IBM Eagle heavy-hex coupling map, with observable \(Z_{62}\) and reference state \(\ket{0}^{\otimes127}\).
The circuit has 4 to 12 layers, and each layer applies 127 single-qubit \(X\) rotations followed by \(ZZ\) rotations on the 144 coupled pairs, with an independent parameter for each gate.
Graph construction uses angles drawn uniformly from \([-\pi,\pi)\) with a fixed seed, Pauli-weight cutoff \(d=6\), and no coefficient truncation.

\section{Benchmark setup and software versions}\label{app:versions}
The codes of Fig.~\ref{fig:inverse-scaling} were measured on one exclusive node of Deucalion x86 (AMD EPYC 7742, 128 cores), each run single-threaded and pinned to its own core.
Up to eight runs shared the node, one per NUMA domain, so that no two shared a last-level cache or memory controller.
Runs timed alone and beside seven others agreed to within the run-to-run spread.
\code{monoprop} is version \code{0.9.3}, \code{ppvm} is \code{0.1.0+git.2570637b1459}, \code{PauliPropagation.jl} is v0.9.0, and \code{MajoranaPropagation.jl} is v0.4.1 with the \code{PauliPropagation.jl} v0.8.2 it requires, on Julia 1.12.7.
Every code is timed identically: each point is the fastest of up to ten repetitions on a freshly built operator.
\label{app:held}We report each run time divided by \(\sum_g K_g\), where \(K_g\) is the number of monomials held when gate \(g\) is applied, giving the mean time per gate per held monomial.
As the first repetition in a process can take up to twice as long, we run a second whenever two fit within the limit, and more while the timed total stays below one minute.
Each point is allowed 30 minutes of timed propagation.
Points exceeding this are omitted, and a code that exceeds it is not run at a deeper circuit of the same \(N\), nor at the next \(N\) of the same depth.
The codes agree exactly on the number of tracked monomials at every point, and on the expectation value to within \(1.2\times10^{-13}\).

\code{PauliPropagation.jl} provides two containers for a sum of Pauli terms, and we propagate the \code{VectorPauliSum} rather than the dictionary-backed \code{PauliSum}.
It holds the terms as unsigned integers and their coefficients in two parallel vectors, maintaining a sorted, duplicate-free prefix into which an unsorted tail is periodically merged.
On the kicked-Ising benchmark it is \(2.7\)--\(5.6\times\) faster than \code{PauliSum} for \(N=16\) to \(64\), and \code{MajoranaPropagation.jl}'s \code{VectorMajoranaSum} is \(2.1\)--\(2.3\times\) faster than its \code{MajoranaSum} on the Hubbard benchmark, so both comparisons take the reference library at its more favourable representation.

\end{document}